\documentclass[12pt,a4paper]{article}

\usepackage{amsmath,amsthm,latexsym,amssymb}
\usepackage{bm,fancybox,multirow,hhline,indentfirst,ulem}
\usepackage[dvipdfmx]{graphicx,overpic}
\usepackage{ascmac,atbegshi}
\usepackage{natbib,verbatim,color}
\usepackage[colorlinks,citecolor=blue,urlcolor=blue]{hyperref}

\makeatletter
\def\section{\@startsection {section}{1}{\z@}{-3.5ex plus -1ex minus -.2ex}{2.3 ex plus .2ex}{\large\sc}}
\def\subsection{\@startsection {subsection}{1}{\z@}{-3.5ex plus -1ex minus -.2ex}{2.3 ex plus .2ex}{\normalsize}}
\makeatother

\theoremstyle{definition}
\newtheorem{thm}{Theorem}
\newtheorem{lem}{Lemma}
\newtheorem{asm}{Assumption}

\newcommand{\argsup}{\mathop{\rm argsup}}
\newcommand{\dlim}{\mathop{\rm dlim}}

\renewcommand{\baselinestretch}{1.2}\selectfont

\allowdisplaybreaks

\def\T{{\rm T}}
\def\E{{\rm E}}
\def\V{{\rm V}}
\def\P{{\rm P}}
\def\o{{\rm o}}
\def\O{{\rm O}}
\def\oP{{\rm o}_{\rm P}}
\def\OP{{\rm O}_{\rm P}}
\def\N{{\rm N}}
\def\U{{\rm U}}

\def\~{\hspace{-1mm}}

\title{\large\bf Akaike information criterion for segmented regression models}
\author{\normalsize Kazuki Nakajima\\\small Department of Statistical Science, The Graduate University for Advanced Studies\bigskip\\\normalsize Yoshiyuki Ninomiya\\\small Department of Fundamental Statistical Mathematics, The Institute of Statistical Mathematics\\\small Department of Statistical Science, The Graduate University for Advanced Studies}
\date{}

\begin{document}

\maketitle

\begin{abstract}
In segmented regression, when the regression function is continuous at the change-points that are the boundaries of the segments, it is also called joinpoint regression, and the analysis package developed by \cite{KimFFM00} has become a standard tool for analyzing trends in longitudinal data in the field of epidemiology.
In addition, it is sometimes natural to expect the regression function to be discontinuous at the change-points, and in the field of epidemiology, this model is used in \cite{JiaZS22}, which is considered important due to the analysis of COVID-19 data.
On the other hand, model selection is also indispensable in segmented regression, including the estimation of the number of change-points; however, it can be said that only BIC-type information criteria have been developed.
In this paper, we derive an information criterion based on the original definition of AIC, aiming to minimize the divergence between the true structure and the estimated structure.
Then, using the statistical asymptotic theory specific to the segmented regression, we confirm that the penalty for the change-point parameter is 6 in the discontinuous case.
On the other hand, in the continuous case, we show that the penalty for the change-point parameter remains 2 despite the rapid change in the derivative coefficients.
Through numerical experiments, we observe that our AIC tends to reduce the divergence compared to BIC.
In addition, through analyzing the same real data as in \cite{JiaZS22}, we find that the selection between continuous and discontinuous using our AIC yields new insights and that our AIC and BIC may yield different results.
\vskip\baselineskip
\vskip\baselineskip

\end{abstract}

\section{Introduction}\label{sec1}

Segmented regression, also known as piecewise regression, can be regarded as a type of change-point analysis and has been used extensively for a long time.
When the regression function is continuous at the change-points that are the boundaries of the segments, it is also called joinpoint regression. 
The analysis package developed by \cite{KimFFM00} has become a standard tool to investigate trends in time series data in the field of medical epidemiology.
For example, it is used in annual statistical reports of large-scale projects analyzing central nervous system tumors (the latest is \citealt{Ost23}), and it is also used in analyses of medically important topics such as cancer and cardiovascular risk (\citealt{How20} and \citealt{He21}).
Although the slope of the regression function changes in joinpoint regression, it is natural to assume that the value of the regression function itself changes, that is, the regression function is discontinuous at the change-points, depending on the situation.
In fact, in the field of epidemiology, this model is used for analysis in \cite{JiaZS22}, which is treated as a discussion paper in the Journal of the Royal Statistical Society Series B.
The data analyzed in that paper are the number of new COVID-19 infections, and provide estimated regression functions like the red line in the right panel of Figure \ref{fig1}.
On the other hand, it is natural to consider the regression function to be continuous when dealing with the number of new infections, and the purpose of this paper is to estimate such a regression function using a statistically valid method.
Actually, in our method, the red line in the left panel is provided, and it is inferred that some of the jumps in the right panel may not exist.

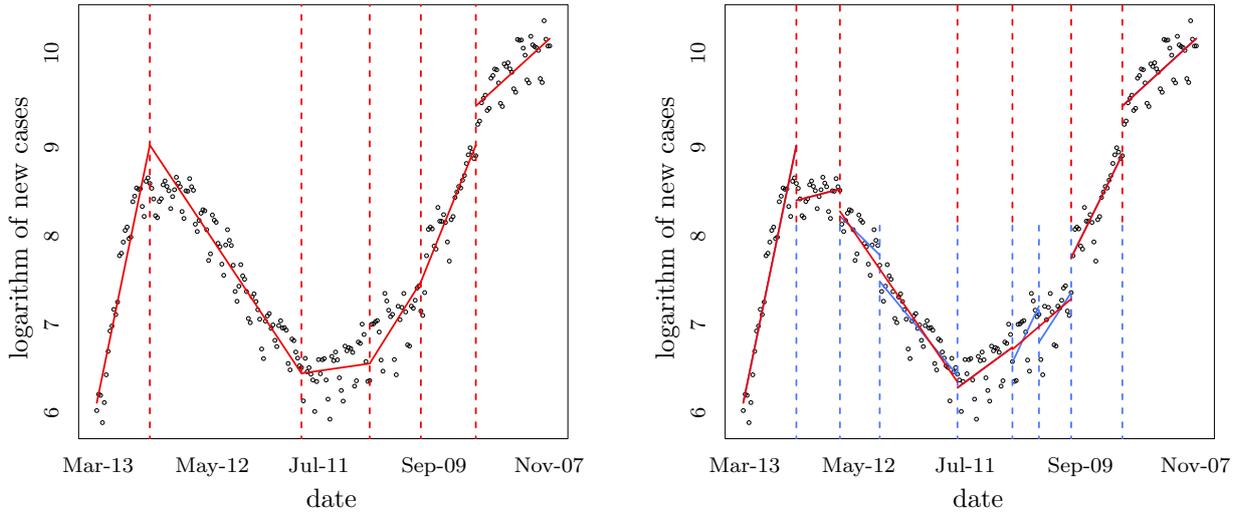
\begin{figure}[!t]
\begin{overpic}[scale=0.44]{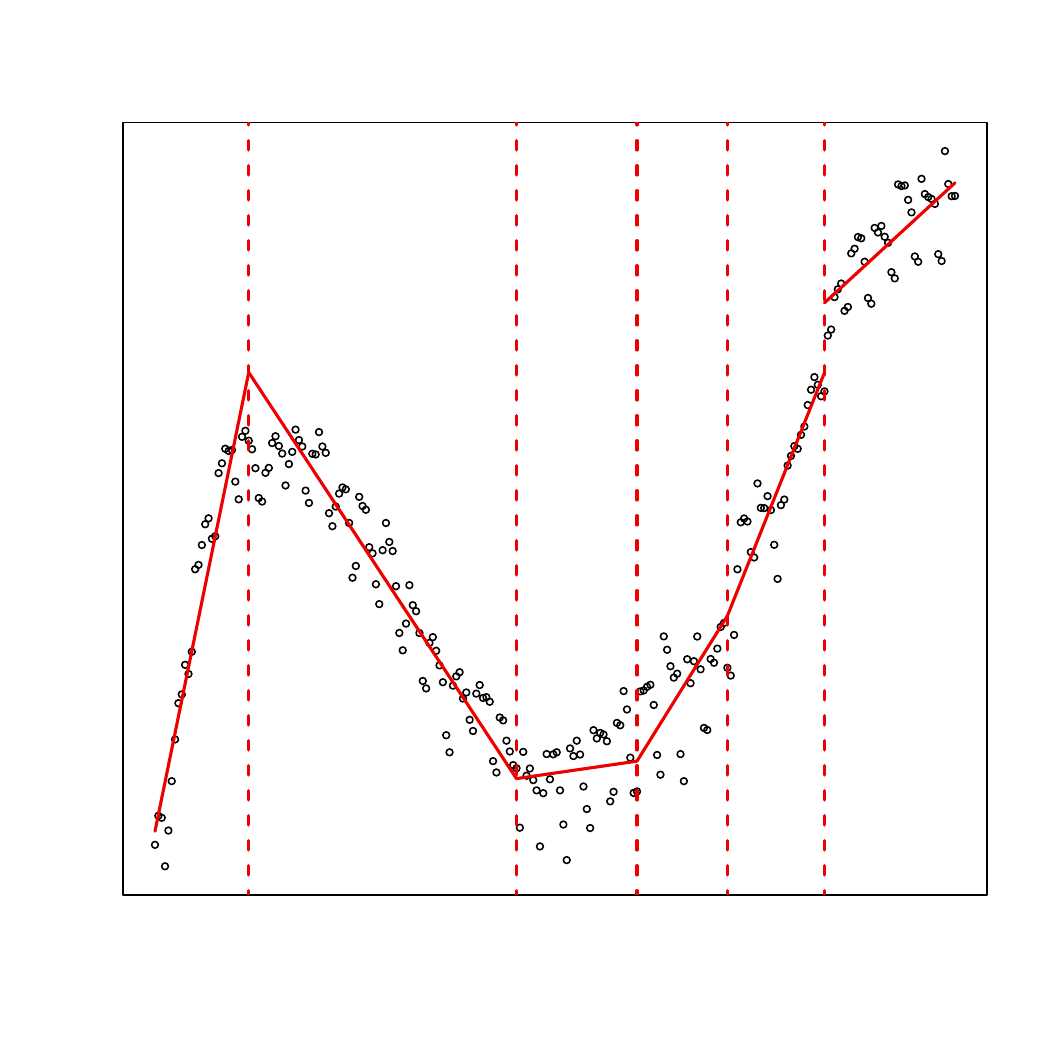}
\put(0,28){\footnotesize\rotatebox{90}{logarithm of new cases}}
\put(6,18){\scriptsize\rotatebox{90}{6}}
\put(6,33){\scriptsize\rotatebox{90}{7}}
\put(6,48){\scriptsize\rotatebox{90}{8}}
\put(6,63){\scriptsize\rotatebox{90}{9}}
\put(6,78){\scriptsize\rotatebox{90}{10}}
\put(9,9){\scriptsize Mar-13}
\put(28,9){\scriptsize May-12}
\put(47,9){\scriptsize Jul-11}
\put(66,9){\scriptsize Sep-09}
\put(85,9){\scriptsize Nov-07}
\put(50,3){\footnotesize date}
\end{overpic}
\hspace{5mm}
\begin{overpic}[scale=0.44]{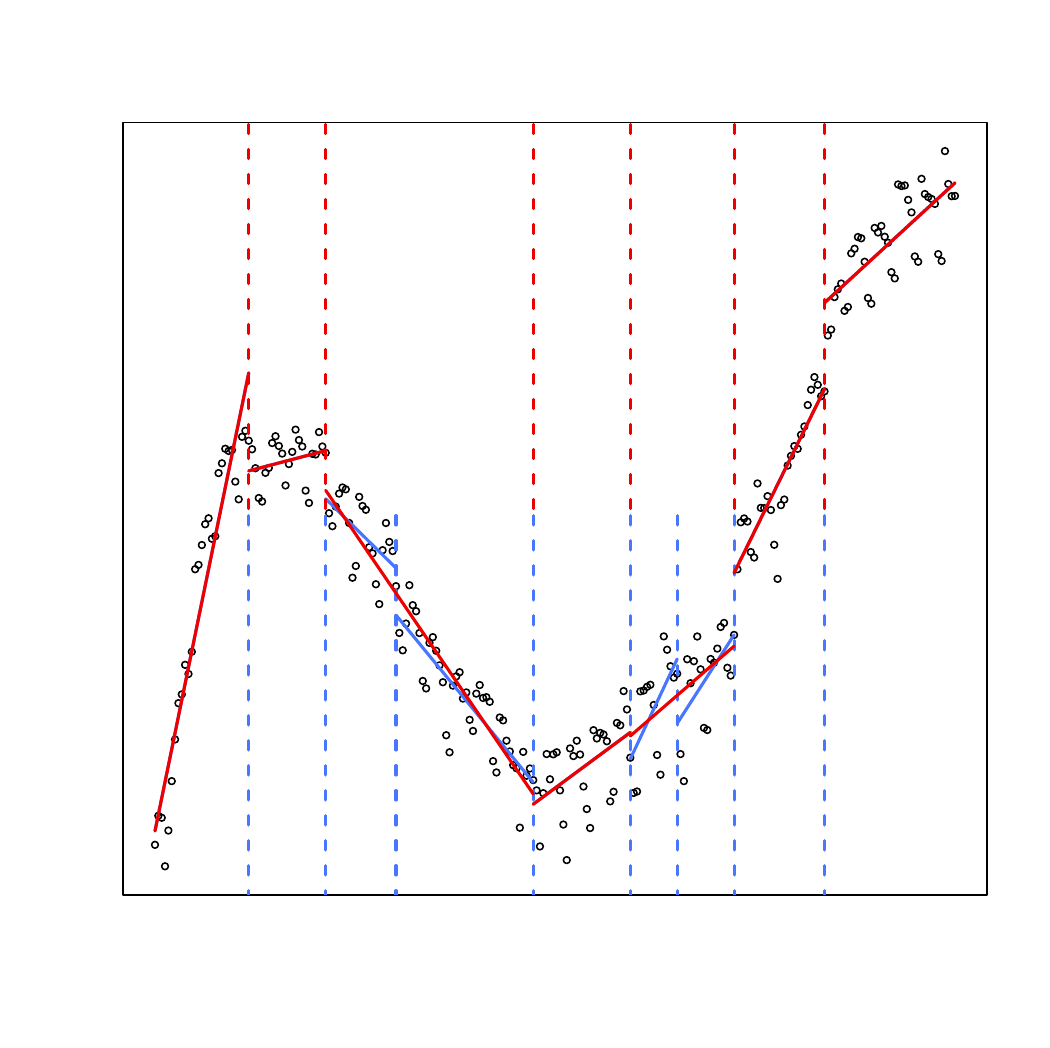}
\put(0,28){\footnotesize\rotatebox{90}{logarithm of new cases}}
\put(6,18){\scriptsize\rotatebox{90}{6}}
\put(6,33){\scriptsize\rotatebox{90}{7}}
\put(6,48){\scriptsize\rotatebox{90}{8}}
\put(6,63){\scriptsize\rotatebox{90}{9}}
\put(6,78){\scriptsize\rotatebox{90}{10}}
\put(9,9){\scriptsize Mar-13}
\put(28,9){\scriptsize May-12}
\put(47,9){\scriptsize Jul-11}
\put(66,9){\scriptsize Sep-09}
\put(85,9){\scriptsize Nov-07}
\put(50,3){\footnotesize date}
\end{overpic}
\caption{Results obtained by allowing the regression function to be continuous at the change-points (left panel) and not allowing it to be continuous (right panel) for the number of new COVID-19 infections in the United Kingdom: The red solid and broken lines are the regression function and change-points estimated using the proposed AIC. The blue solid and broken lines in the right panel are the regression function and change-points estimated using AIC$_{\rm naive}$.}
\label{fig1}
\end{figure}

In the above papers, denoting the response variable as $y$ and the explanatory variable as $x$, models given by
\begin{align}
y_i=\sum_{k=1}^{m+1}I_{(\tau^{[k-1]},\tau^{[k]}]}(x_i)(\theta_1^{[k]}+\theta_2^{[k]}x_i)+\epsilon_i, \qquad i\in\{1,2,\ldots,n\}
\label{model}
\end{align}
are treated, where $I_A(\cdot)$ is the indicator function for $A\subset\mathbb{R}$, and $I_A(x)=1$ when $x\in A$ and $I_A(x)=0$ when $x\notin A$.
In addition, $\tau^{[0]}$ and $\tau^{[m+1]}$ are known constants, $\bm{\tau}=(\tau^{[1]},\tau_2,\ldots,\tau^{[m]})^\T$ and $\bm{\theta}^{[k]}=(\theta_1^{[k]}, \theta_2^{[k]})^{\T}$ ($k=1,2,\ldots,m+1$) are unknown parameters.
This model has a regression function for each segment $(\tau^{[k-1]},\tau^{[k]}]$, and the change-point $\tau^{[k]}$ is also called a node.
If $\theta_1^{[k]}+\theta_2^{[k]}\tau^{[k]}=\theta_1^{[k+1]}+\theta_2^{[k+1]}\tau^{[k]}$, then the regression function is continuous at $\tau^{[k]}$; and if $\theta_1^{[k]}+\theta_2^{[k]}\tau^{[k]}\neq\theta_1^{[k+1]}+\theta_2^{[k+1]}\tau^{[k]}$, the regression function is discontinuous at $\tau^{[k]}$.

Segmented regression has long been of interest not only from an applied perspective, but also from a theoretical perspective.
This is because the change-point parameter $\tau^{[k]}$ requires specific care in statistical asymptotic theory. 
For example, let us suppose that the regression function is discontinuous at $\tau^{[k]}$ and that $\epsilon_i$ follows some parametric distribution independently.
In this case, the log-likelihood function is not differentiable at $\tau^{[k]}$ and therefore lacks regularity, leading to the development of a specific asymptotic theory.
In fact, since Taylor expansion is not applicable, the maximum likelihood estimator of $\tau^{[k]}$ does not have asymptotic normality.
However, when the regression function is continuous at $\tau^{[k]}$, the maximum likelihood estimator eventually satisfies asymptotic normality, as shown in \cite{Fed75a}.
The proof method is specific to segmented regression and for some time there were few extensions of the theory; \cite{Wu92} addressed the case where the covariates are multidimensional, and the final form was given by \cite{LiuWZ97}.
In addition, \cite{KimK08} deals with cases where there are two variables to consider for segmentation, and what is noteworthy is the sophistication of the proof method, which we will rely on in this paper.
The corresponding author of this paper has developed segmented regression in various directions, not limited to those in \cite{KimFFM00} and \cite{KimK08}, which are summarized in \cite{KimCBWF22}.
Regarding the specific asymptotic theory when the regression function is discontinuous at $\tau^{[k]}$, Chapter 3 of \cite{CsoH97} deals with this topic; the important papers that first addressed this topic are \cite{GomH94}, \cite{Bai97}, and \cite{BaiP98}.
There are numerous extensions to these methods; notable examples include \cite{AueH13}, which reviews cases with serial correlation, \cite{BerGHK09}, which begins the extension to function data, and \cite{AueHHR09}, which handles high-dimensional data with an infinite variance-covariance structure.
Recently, segmented regression has also gained importance in the field of machine learning, with algorithm development being carried out in \cite{AchDLS16}, \cite{SiaGKS20}, \cite{Bem22}, and others.

In segmented regression, it is natural to assume that the number of change-points is unknown, just as it is usually assumed that the locations of change-points are unknown.
Estimating this number is part of model selection, indicating that model selection is an essential task even in segmented regression.
In fact, the aforementioned \cite{LiuWZ97} proposes an information criterion that can be called BIC-type, which is slightly more penalizing than BIC (\citealt{Sch78}), and shows that model selection consistency is achieved regardless of whether the regression function is continuous or discontinuous.
In addition, \cite{KimYF09} compares test-based methods and several information criteria and concludes that BIC is superior in terms of the probability of selecting the true model.
These BIC-types impose the same penalty on the change-points regardless of whether the regression function is continuous or discontinuous, but this is unnatural considering that the asymptotic properties of the estimators differ between the two.
Therefore, \cite{KimK16} refers to \cite{ZhaS07}, which derived an information criterion for fundamental change-point models based on the original derivation of BIC, and intuitively proposes different penalty terms for the two, showing that the consistency of model selection is maintained.
In this paper, with the aim of theoretically providing different penalty terms for continuous and discontinuous regression functions, we derive an information criterion based on the original derivation of AIC (Aka73), that is, as an asymptotic bias-corrected estimator of the Kullback-Leibler divergence.
Note that the target of model selection differs from that of BIC, so the goal is not to increase the probability of selecting the true model but to minimize the divergence between the true and estimated structures.
When the regression function is discontinuous, the method of \cite{Nin15}, which derived AIC for the basic change-point model, can be used. As a result, a penalty of $6$, which is much larger than the penalty of $2$ that appears for each parameter in AIC of regular statistical models, appears for the change-point parameter.
On the other hand, when the regression function is continuous, the penalty for the change-point parameter remains $2$.

This paper is organized as follows.
In Section \ref{sec2}, the model to be treated is provided first as a preparation.
There, rather than treating the target of this paper as in \eqref{model}, we consider a natural generalization that is based on generalized linear models or not necessarily piecewise linear.
We then confirm that the maximum likelihood estimator is consistent.
In Section \ref{sec3}, we first show that the estimators for the continuous segmented regression model are asymptotically normal.
Since the model is not necessarily piecewise linear, it is not possible to show this in the same way as in the method of \cite{KimK08}, and some assumptions are necessary, although the result will be of the same type.
We then derive AIC based on it and see that the resulting penalty for the change-point parameter is $2$.
This AIC is used intuitively (e.g. \citealt{KimCBWF22}), and our contribution is to give its guarantee in a generalized model.
Since the derivative of the regression function changes abruptly even when the regression function is continuous, one might intuitively think that the penalty for the change-point parameter is $6$, and a naive AIC would define it as such.
In Section \ref{sec4}, we derive AIC directly for the discontinuous segmented regression model, omitting to mention that the asymptotic properties of the estimators are not different from those of the fundamental change-point model, since it is trivial.
As a result, we see that the penalty for the change-point parameter is $6$.
As in \cite{Nin15}, the penalty is set to $2$ for the naive AIC.
In Section \ref{sec5}, we first check through numerical experiments that the derived penalty term accurately approximates the bias that appears when evaluating the Kullback-Leibler divergence.
Then, model selection by the derived AIC, naive AIC, and BIC is performed on artificial data, and the results are compared to confirm that the derived AIC gives successful results in terms of reducing the divergence.
In Section \ref{sec6}, a real data analysis is performed to obtain Figure \ref{fig1}.
We also check that the results of these three criteria are reasonably different.
Section \ref{sec7} gives a summary and discussion.

\section{Preparation}\label{sec2}

For each sample $i\in\{1,2,\ldots,n\}$, let us suppose that the explanatory and response variable pairs $(x_i,y_i)$ are obtained and $(x_i,y_i)$ independently follow the same distribution as $(x,y)$.
Here, $x$ is a random variable following a known probability function $q(\cdot)$ taking the value in $(\tau^{[0]},\tau^{[m+1]}]$, $y$ is a random variable following a generalized linear model that differs for each segment to which $x$ belongs, and the simultaneous probability function is
\begin{align}
f(x,y;\bm{\xi})=\sum_{k=1}^{m+1}I_{(\tau^{[k-1]},\tau^{[k]}]}(x)\exp\{y\phi(x;\bm{\theta}^{[k]})-a(\phi(x;\bm{\theta}^{[k]}))+b(y)\} q(x),
\label{model2}
\end{align}
where the space of parameter $\bm{\theta}^{[k]}\in\mathbb{R}^p$ is compact, $a(\cdot)$, $b(\cdot)$ and $\phi(\cdot)$ are known C$^2$ class functions, $\tau^{[0]}$ and $\tau^{[m+1]}$ are known values.
Setting $\bm{x}=(1,x,x^2,\ldots,x^{p-1})^{\T}$, the example supposed in $\phi(x;\bm{\theta}^{[k]})$ is $\bm{\theta}^{\T}\bm{x}$.
Also, often later, we use the notation $\bm{\theta}=(\bm{\theta}^{[1]\T},\bm{\theta}^{[2]\T},\ldots,\allowbreak\bm{\theta}^{[m+1]\T})^\T$, $\bm{\tau}=(\tau^{[1 ]},\tau^{[2]},\ldots,\tau^{[m]})^{\T}$, $\bm{\xi}=(\bm{\theta}^\T,\bm{\tau}^\T)^\T$.
Needless to say, this model is a generalization of \eqref{model}.

Let us denote the true values of parameters by attaching $^*$, so that, for example, the true value of $\bm{\xi}$ is denoted by $\bm{\xi}^*$.
We define $\tau^{[0]*}$ and $\tau^{[m+1]*}$ as the known values of $\tau^{[0]}$ and $\tau^{[m+1]}$, respectively.
For this $\bm{\xi}^*$, let us consider the maximum likelihood estimator $\hat{\bm{\xi}}=\argsup_{\bm{\xi}}(\prod_{i=1}^n f(x_i,y_i;\bm{\xi}))$.
Denoting the difference between the log-likelihood function for $(x,y)$ and the true log-likelihood as
\begin{align*}
g(\bm{\xi};x,y) = & \sum_{k=1}^{m+1}I_{(\tau^{[k-1]},\tau^{[k]}]}(x)\{y\phi(x;\bm{\theta}^{[k]})-a(\phi(x;\bm{\theta}^{[k]}))\} 
\\
& - \sum_{k=1}^{m+1}I_{(\tau^{[k-1]*},\tau^{[k]*}]}(x)\{y\phi(x;\bm{\theta}^{[k]*})-a(\phi(x;\bm{\theta}^{[k]*}))\},
\end{align*}
we can write that $\hat{\bm{\xi}}=\argsup_{\bm{\xi}}(\sum_{i=1}^n g(\bm{\xi};x_i, y_i)/n)$.
In order to discuss the asymptotic properties of $\hat{\bm{\xi}}$, we make the following assumption.
\begin{asm}
In any small neighborhood of $\tau^{[k]*}$, $x$ has a continuous and positive probability density function for $k\in\{1,2,\ldots,m\}$.
In addition, the regression function or its derivative is not continuous at the change-points, that is, $\phi(\tau^{[k]*}; \bm{\theta}^{[k]*})\neq\phi(\tau^{[k]*};\bm{\theta}^{[k+1]*})$ or $\partial\phi(\tau^{[k]*};\bm{\theta}^{[k]*})/\partial x\neq\partial\phi(\tau^{[k]*};\bm{\theta}^{[k+1]*})/\partial x$ for $k\in\{1,2,\ldots,m\}$.
\label{asm1}
\end{asm}
\noindent Then the following holds.
\begin{thm}
Under Assumption \ref{asm1}, $\hat{\bm{\xi}}$ is consistent, i.e. $\P(\|\hat{\bm{\xi}}-\bm{\xi}^*\|\leq\delta)=\o(1)$ for any $\delta>0$.
\label{thm1}
\end{thm}
\begin{proof}
Taking $\sup_{\|\bm{\xi}-\bm{\xi}^*\|>\delta}$ in $n^{-1}\sum_{i=1}^ng(\bm{\xi};x_i,y_i) \leq \E[g(\bm{\xi};x,y)]+||n^{-1}\sum_{i=1}^ng(\bm{\xi};x_i,y_i)-\E[g(\bm{\xi};x,y)]|$, which can be easily checked, we obtain
\begin{align}
& \sup_{\|\bm{\xi}-\bm{\xi}^*\|>\delta}\frac{1}{n}\sum_{i=1}^ng(\bm{\xi};x_i,y_i)
\notag \\
& \leq \sup_{\|\bm{\xi}-\bm{\xi}^*\|>\delta}\E[g(\bm{\xi};x,y)] + \sup_{\|\bm{\xi}-\bm{\xi}^*\|>\delta}\bigg|\frac{1}{n}\sum_{i=1}^ng(\bm{\xi};x_i,y_i)-\E[g(\bm{\xi};x,y)]\bigg|.
\label{th1ineq}
\end{align}
To prove consistency, it is sufficient to show that the left-hand side is smaller than $n^{-1}\sum_{i=1}^ng(\bm{\xi}^*;\allowbreak x_i,y_i)=0$ with probability approaching $1$ as $n\to\infty$.
First, let us consider the first term on the right-hand side of \eqref{th1ineq}.
Letting $\mu(\bm{\xi};x)\equiv\sum_{k=1}^{m+1}I_{(\tau^{[k-1]},\tau^{[k]}]}(x)\phi(x;\bm{\theta}^{[k]})$, there exists $\epsilon>0$ such that we can take $c>0$ satisfying
\begin{align*}
\inf_{\|\bm{\xi}-\bm{\xi}^*\|>\delta}\P(|\mu(\bm{\xi},x)-\mu(\bm{\xi}^*,x)|>c\|\bm{\xi}-\bm{\xi}^*\|)>\epsilon
\end{align*}
for any $\delta>0$. 
Now we use 
\begin{align*}
\E[g(\bm{\xi};x,y)] & = \E[y\{\mu(\bm{\xi};x)-\mu(\bm{\xi}^*;x)\}-\{a(\mu(\bm{\xi};x))-a(\mu(\bm{\xi}^*;x))\}]
\\
& = \E[a'(\mu(\bm{\xi}^*;x))\{\mu(\bm{\xi};x)-\mu(\bm{\xi}^*;x)\}-\{a(\mu(\bm{\xi};x))-a(\mu(\bm{\xi}^*;x))\}]
\\
& = \E[-a''(\tilde{\mu}(\bm{\xi},\bm{\xi}^*;x))\{\mu(\bm{\xi};x)-\mu(\bm{\xi}^*;x)\}^2],
\end{align*}
where $a'(\cdot)$ and $a''(\cdot)$ denote the first and second derivatives of $a(\cdot)$, respectively, and $\tilde{\mu}(\bm{\xi},\bm{\xi}^*;x)$ is an appropriate value between $\mu(\bm{\xi};x)$ and $\mu(\bm{\xi}^*;x)$.
Since it follows from the properties of $a(\cdot)$ that $\eta\equiv\inf_{\bm{\xi},x}a''(\tilde{\mu}(\bm{\xi},\bm{\xi}^*;x))>0$, we obtain
\begin{align*}
\E[g(\bm{\xi};x,y)] \leq -\E[\eta(c\delta)^2I_{\{|\mu(\bm{\xi},x)-\mu(\bm{\xi}^*,x)|\ge c\delta\}}(x)] \leq -\eta c^2\delta^2\epsilon
\end{align*}
for any $\bm{\xi}$ such that $\|\bm{\xi}-\bm{\xi}^*\|>\delta$.
Next, regarding the second term on the right-hand side in \eqref{th1ineq}, Lemma 19.34 of \cite{Van00}, for example, indicates that it is $\oP(1)$. 
This is because $g(\bm{\xi};x,y)$ is differentiable with respect to $\bm{\xi}$ and the parameter space is compact and therefore the bracketing integral is bounded.
This completes the proof.
\end{proof}

\noindent
In the following sections, we will derive AIC based on asymptotics; from the consistency, $\hat{\tau}^{[k^{\dagger}]}<\hat{\tau}^{[k^{\ddagger}]}$ can be assumed if $\tau^{[k^{\dagger}]*}<\tau^{[k^{\ddagger}]*}$. 
In addition, when it is necessary to determine the order of $\tau^{[k]*}$ and $\tau^{[k]}$ in the proof, we assume that $\tau^{[k]*}<\tau^{[k]}$, and give the result without this assumption.
This is because although the case $\tau^{[k]*}>\tau^{[k]}$ can be trivially developed in the same way, describing both cases would be unnecessarily tedious. 

\section{AIC for continuous segmented regression models}\label{sec3}

To discuss the asymptotic properties of the estimators, it is necessary to distinguish cases where the regression function $\sum_{k=1}^{m+1}I_{(\tau^{[k-1]},\tau^{[k]}]}(x)\phi(x;\bm{\theta}^{[k]})$ is continuous or not at the change-points $\tau^{[1]},\tau^{[2]},\ldots ,\tau^{[m]}$, and we consider the continuous case in this section.
As in \cite{KimK08}, letting
\begin{align*}
& G(\bm{\theta}^{\dagger},\bm{\theta}^{\ddagger},(\tau^{\dagger},\tau^{\ddagger}])
\\
& \equiv\frac{1}{n}\sum_{i=1}^{n}I_{(\tau^{\dagger},\tau^{\ddagger}]}(x_i)[\{y_i\phi(x_i;\bm{\theta}^{\dagger})-a(\phi(x_i;\bm{\theta}^{\dagger}))\} - \{y_i\phi(x_i;\bm{\theta}^{\ddagger})-a(\phi(x_i;\bm{\theta}^{\ddagger}))\}]
\end{align*}
for any $\bm{\theta}^{\dagger},\bm{\theta}^{\ddagger}\in\mathbb{R}^p$ and $(\tau^{\dagger},\tau^{\ddagger}]\subset\mathbb{R}$, we have
\begin{align*}
\frac{1}{n}\sum_{i=1}^ng(\bm{\xi};x_i, y_i)=\sum_{k=1}^{m+1}G(\bm{\theta}^{[k]},\bm{\theta}^{[k]*},(\tau^{[k-1]},\tau^{[k]*}])+\sum_{k=1}^{m}G(\bm{\theta}^{[k]},\bm{\theta}^{[k+1]*},(\tau^{[k]*},\tau^{[k]}]). 
\end{align*}
Further letting
\begin{align*}
& H(\bm{\theta}) \equiv \sum_{k=1}^{m+1}G(\bm{\theta}^{[k]},\bm{\theta}^{[k]*},(\tau^{[k-1]*},\tau^{[k]*}]),
\\
& R(\bm{\theta},\bm{\tau}) \equiv \frac{1}{n}\sum_{i=1}^ng(\bm{\xi};x_i, y_i) - H(\bm{\theta})
\\
& = \sum_{k=1}^{m}(G(\bm{\theta}^{[k]},\bm{\theta}^{[k+1]*},(\tau^{[k]*},\tau^{[k]}]) - G(\bm{\theta}_{k+1},\bm{\theta}^{[k+1]*},(\tau^{[k]*},\tau^{[k]}])), 
\end{align*}
the following two lemmas hold.
The first is fundamental and not specific to this setting, and therefore the proof is given in Appendix.
\begin{lem}
Let $\bm{w}$ be a random vector following a multi-dimensional Gaussian distribution ${\rm N}(\bm{0}_{p(m+1)},\bm{T})$ with the mean vector $\bm{0}_{p(m+1)}$ and the variance-covariance matrix $\bm{T}$, where
\begin{align*}
\bm{T}^{[k]} = \E\bigg[I_{(\tau^{[k-1]*},\tau^{[k]*}]}(x)a''(\phi(x;\bm{\theta}^{[k]*}))\bigg\{\frac{\partial}{\partial\bm{\theta}}\phi(x;\bm{\theta}^{[k]*})\bigg\}\bigg\{\frac{\partial}{\partial\bm{\theta}}\phi(x;\bm{\theta}^{[k]*})\bigg\}^\T\bigg]
\end{align*}
and
\begin{align*}
\bm{T} = {\rm diag}(\bm{T}^{[1]},\bm{T}^{[2]},\ldots,\bm{T}^{[m+1]}) = 
\begin{pmatrix}
\bm{T}^{[1]} & \bm{\mbox{O}}_{p\times p} & \cdots & \bm{\mbox{O}}_{p\times p} \\
\bm{\mbox{O}}_{p\times p} & \bm{T}^{[2]} & \ddots & \vdots \\
\vdots & \ddots & \ddots & \bm{\mbox{O}}_{p\times p} \\
\bm{\mbox{O}}_{p\times p} & \cdots & \bm{\mbox{O}}_{p\times p} & \bm{T}^{[m+1]} \\
\end{pmatrix}.
\end{align*}
Then, under Assumption \ref{asm1}, if $\|\bm{\theta}-\bm{\theta}^*\|=\o(1)$, it holds that
\begin{align*}
H(\bm{\theta}) = -\frac{1}{2}(\bm{\theta}-\bm{\theta}^*)^\T\{\bm{T}+\oP(1)\}(\bm{\theta}-\bm{\theta}^*)+\frac{1}{\sqrt{n}}(\bm{\theta}-\bm{\theta}^*)^\T\{\bm{w}+\oP(1)\}+\oP(\|\bm{\theta}-\bm{\theta}^*\|^2).
\end{align*}
\label{lem1}
\end{lem}

\noindent 
The second is specific to this setting and differs from \cite{KimK08} in some respects, and therefore the proof is included in this section.

\begin{lem}
Under Assumption \ref{asm1}, if $\|\bm{\theta}-\bm{\theta}^*\|=\o(1)$ and $\|\bm{\tau}-\bm{\tau}^*\|=\O(\|\bm{\theta}-\bm{\theta}^*\|)$, it holds that
\begin{align*}
R(\bm{\theta},\bm{\tau}) = \oP(\|\bm{\theta}-\bm{\theta}^*\|^2) + \oP\bigg(\frac{1}{\sqrt{n}}\|\bm{\theta}-\bm{\theta}^*\|\bigg).
\end{align*}
\label{lem2}
\end{lem}

\begin{proof}
Let us evaluate
\begin{align*}
& G(\bm{\theta}^{[k]},\bm{\theta}^{[k+1]*},(\tau^{[k]*},\tau^{[k]}])
\\
& = \frac{1}{n}\sum_{k=1}^{m+1}\sum_{i=1}^nI_{(\tau^{[k]*},\tau^{[k]}]}(x_i)(\phi(x_i;\bm{\theta}^{[k]})-\phi(x_i;\bm{\theta}^{[k+1]*}))(y_i-a'(\phi(x_i;\bm{\theta}^{[k+1]*})))
\\
& \ \phantom{=} - \frac{1}{2n}\sum_{k=1}^{m+1}\sum_{i=1}^nI_{(\tau^{[k]*},\tau^{[k]}]}(x_i)(\phi(x_i;\bm{\theta}^{[k]})-\phi(x_i;\bm{\theta}^{[k+1]*}))^2a''(\tilde{\phi}(x_i;\bm{\theta}^{[k]},\bm{\theta}^{[k+1]*})),
\end{align*}
which is obtained from \eqref{forlem2} in Appendix.
First, regarding the first term, the continuity of the regression function at the change-points indicates that 
\begin{align}
& |\phi(x_i;\bm{\theta}^{[k]})-\phi(x_i;\bm{\theta}^{[k+1]*})|
\notag \\
& \le |\phi(\tau^{[k]*};\bm{\theta}^{[k]})-\phi(\tau^{[k]*};\bm{\theta}^{[k]*})| + |\phi(x_i;\bm{\theta}^{[k]})-\phi(\tau^{[k]*};\bm{\theta}^{[k]})|
\notag \\
& \ \phantom{\le} + |\phi(x_i;\bm{\theta}^{[k+1]*})-\phi(\tau^{[k]*};\bm{\theta}^{[k+1]*})|
\notag \\
& = \O(\|\bm{\theta}^{[k]}-\bm{\theta}^{[k]*}\|) + \O(|\tau^{[k]}-\tau^{[k]*}|) + \O(|\tau^{[k]}-\tau^{[k]*}|)
\notag \\
& = \O(\|\bm{\theta}-\bm{\theta}^*\|)
\label{fromconti}
\end{align}
when $x_i\in(\tau^{[k]*},\tau^{[k]}]$.
In addition, the expectation of 
\begin{align*}
\frac{1}{\sqrt{n}}\sum_{i=1}^nI_{(\tau^{[k]*},\tau^{[k]}]}(x_i)\frac{\phi(x_i;\bm{\theta}^{[k]})-\phi(x_i;\bm{\theta}^{[k+1]*})}{\|\bm{\theta}-\bm{\theta}^*\|}(y_i-a'(\phi(x_i,\bm{\theta}^{[k+1]*})))
\end{align*}
is obviously $0$.
Therefore, using $\tau^{[k]}-\tau^{[k]*}=\o(1)$ and \eqref{fromconti}, Lemma 19.34 of \cite{Van00} implies that this is uniformly $\oP(1)$.
As a result, we obtain
\begin{align*}
& \frac{1}{n}\sum_{k=1}^{m+1}\sum_{i=1}^nI_{(\tau^{[k]*},\tau^{[k]}]}(x_i)(\phi(x_i;\bm{\theta}^{[k]})-\phi(x_i;\bm{\theta}^{[k+1]*}))(y_i-a'(\phi(x_i,\bm{\theta}^{[k+1]*})))
\\
& = \oP\bigg(\frac{1}{\sqrt{n}}\|\bm{\theta}-\bm{\theta}^*\|\bigg). 
\end{align*}
Next, regarding the second term, the continuity of the regression function at the change-points indicates that 
\begin{align}
& |\phi(x_i;\bm{\theta}^{[k]})-\phi(x_i;\bm{\theta}^{[k+1]*})|^2
\notag \\
& \le (|\phi(\tau^{[k]*};\bm{\theta}^{[k]})-\phi(\tau^{[k]*};\bm{\theta}^{[k]*})| + |\phi(x_i;\bm{\theta}^{[k]})-\phi(\tau^{[k]*};\bm{\theta}^{[k]})|
\notag \\
& \ \phantom{\le} + |\phi(x_i;\bm{\theta}^{[k+1]*})-\phi(\tau^{[k]*};\bm{\theta}^{[k+1]*})|)^2
\notag \\
& = \O(\|\bm{\theta}-\bm{\theta}^*\|^2)
\label{fromconti2}
\end{align}
when $x_i\in(\tau^{[k]*},\tau^{[k]}]$.
Therefore, it follows from $\tau^{[k]}-\tau^{[k]*}=\o(1)$ and \eqref{fromconti2} that 
\begin{align*}
\frac{1}{n}\sum_{i=1}^nI_{(\tau^{[k]*},\tau^{[k]}]}(x_i)\frac{(\phi(x_i;\bm{\theta}^{[k]})-\phi(x_i;\bm{\theta}^{[k+1]*}))^2}{\|\bm{\theta}-\bm{\theta}^*\|^2}a''(\tilde{\phi}(x_i;\bm{\theta}^{[k]},\bm{\theta}^{[k+1]*})) = \oP(1).
\end{align*}
As a result, we obtain
\begin{align*}
\frac{1}{n}\sum_{k=1}^{m+1}\sum_{i=1}^nI_{(\tau^{[k]*},\tau^{[k]}]}(x_i)(\phi(x_i;\bm{\theta}^{[k]})-\phi(x_i;\bm{\theta}^{[k+1]*}))^2a''(\tilde{\phi}(x_i;\bm{\theta}^{[k]},\bm{\theta}^{[k+1]*})) = \oP(\|\bm{\theta}-\bm{\theta}^*\|^2).
\end{align*}
Thus, it holds that 
\begin{align*}
G(\bm{\theta}^{[k]},\bm{\theta}^{[k+1]*},(\tau^{[k]*},\tau^{[k]}]) = \oP(\|\bm{\theta}-\bm{\theta}^*\|^2) + \oP\bigg(\frac{1}{\sqrt{n}}\|\bm{\theta}-\bm{\theta}^*\|\bigg).
\end{align*}
A similar evaluation can be obtained for 
$G(\bm{\theta}_{k+1},\bm{\theta}^{[k+1]*},(\tau^{[k]*},\tau^{[k]}])$, which completes the proof.
\end{proof}

\noindent
After obtaining the asymptotic expression for the change-point estimator, combining Lemmas \ref{lem1} and \ref{lem2} yields the following theorem.

\begin{thm}
$\sqrt{n}(\hat{\bm{\theta}}-\bm{\theta}^*)$ and $\sqrt{n}(\hat{\bm{\tau}}-\bm{\tau}^*)$ converge in distribution to $\N(\bm{0}_{p(m+1)},\bm{T}^{-1})$ and $\N(\bm{0}_{m},\bm{A}^{\T}\bm{T}^{-1}\bm{A})$, respectively.
\label{thm2}
\end{thm}

\begin{proof}
Since the regression function $\sum_{k=1}^{m+1}I_{(\tau^{[k-1]},\tau^{[k]}]}(x)\phi(x;\bm{\theta}^{[k]})$ is continuous at the change-points $\tau^{[1]},\tau^{[2]},\ldots,\tau^{[m]}$, we have $\phi(\hat{\tau}^{[k]};\hat{\bm{\theta}}^{[k]})=\phi(\hat{\tau}^{[k]};\hat{\bm{\theta}}^{[k+1]})$ for each $k\in\{1,2,\ldots,\allowbreak m\}$.
Since it follows from Theorem \ref{thm1} that $\hat{\tau}^{[k]}-\tau^{[k]*}=\oP(1)$, we obtain
\begin{align*}
0 & = \phi(\tau^{[k]*};\hat{\bm{\theta}}^{[k+1]})-\phi(\tau^{[k]*};\hat{\bm{\theta}}^{[k]})
\\
& \ \phantom{=} +\bigg(\frac{\partial}{\partial x}\phi(\tau^{[k]*};\hat{\bm{\theta}}^{[k+1]})-\frac{\partial}{\partial x}\phi(\tau^{[k]*};\hat{\bm{\theta}}^{[k]})\bigg)(\hat{\tau}^{[k]}-\tau^{[k]*})(1+\oP(1)).
\end{align*}
Therefore, it holds that
\begin{align*}
& \hat{\tau}^{[k]}-\tau^{[k]*}
\\
& = -\frac{\phi(\tau^{[k]*};\hat{\bm{\theta}}^{[k+1]})-\phi(\tau^{[k]*};\hat{\bm{\theta}}^{[k]})}{\partial\phi(\tau^{[k]*};\hat{\bm{\theta}}^{[k+1]})/\partial x-\partial\phi(\tau^{[k]*};\hat{\bm{\theta}}^{[k]})/\partial x}(1+\oP(1))
\\
& = -\frac{(\partial\phi(\tau^{[k]*};\bm{\theta}^{[k+1]*})/\partial\bm{\theta})^{\T}(\hat{\bm{\theta}}^{[k+1]}-\bm{\theta}^{[k+1]*})-(\partial\phi(\tau^{[k]*};\bm{\theta}^{[k]*})/\partial\bm{\theta})^{\T}(\hat{\bm{\theta}}^{[k]}-\bm{\theta}^{[k]*})}{\partial\phi(\tau^{[k]*};\bm{\theta}^{[k+1]*})/\partial x-\partial\phi(\tau^{[k]*};\bm{\theta}^{[k]*})/\partial x}
\\
& \ \phantom{=} \times (1+\oP(1)).
\end{align*}
The second equality comes from $\hat{\bm{\theta}}^{[k]}-\bm{\theta}^{[k]*}=\oP(1)$ in Theorem \ref{thm1}.
Hence, letting
\begin{align*}
\bm{a}_{i,j}\equiv\left\{
\begin{array}{ll}
\displaystyle\bigg(\frac{\partial}{\partial x}\phi(\tau^{[k]*};\bm{\theta}^{[k+1]*})-\frac{\partial}{\partial x}\phi(\tau^{[k]*};\bm{\theta}^{[k]*})\bigg)^{-1}\frac{\partial}{\partial\bm{\theta}}\phi(\tau^{[k]*};\bm{\theta}^{[k]*}) & (i=j=k)
\medskip\\
\displaystyle\bigg(\frac{\partial}{\partial x}\phi(\tau^{[k]*};\bm{\theta}^{[k+1]*})-\frac{\partial}{\partial x}\phi(\tau^{[k]*};\bm{\theta}^{[k]*})\bigg)^{-1}\frac{\partial}{\partial\bm{\theta}}\phi(\tau^{[k]*};\bm{\theta}^{[k+1]*}) & (i-1=j=k)
\medskip\\
\bm{0}_p & ({\rm otherwise})
\end{array}
\right.
\end{align*}
and defining $\{(m+1)p\times m\}$-matrix $\bm{A}$ as $(\bm{a}_{i,j})_{1\le i\le m+1,1\le j\le m}$, it can be written that
\begin{align}
\hat{\bm{\tau}}-\bm{\tau}^*=\bm{A}^{\T}(\hat{\bm{\theta}}-\bm{\theta}^*)(1+\oP(1)).
\label{asymtau}
\end{align}
Note that, from Assumption \ref{asm1}, the term whose inverse is taken in the definition of $\bm{a}_{i,j}$ is not equal to $0$.

Now, from Lemmas \ref{lem1} and Lemma \ref{lem2}, it holds uniformly for $n$ and $\bm{\xi}\in\{\bm{\xi}:\|\bm{\theta}-\bm{\theta}^*\|=\o(1),\ \|\bm{\tau}-\bm{\tau}^*\|=\o(\|\bm{\theta}-\bm{\theta}^*\|)\}$ that
\begin{align}
& \frac{1}{n}\sum_{i=1}^ng(\bm{\xi};x_i,y_i) = H(\bm{\theta})+R(\bm{\theta},\bm{\tau})
\notag \\
& = -\frac{1}{2}(\bm{\theta}-\bm{\theta}^*)^\T(\bm{T}+\oP(1))(\bm{\theta}-\bm{\theta}^*)+\frac{1}{\sqrt{n}}(\bm{\theta}-\bm{\theta}^*)^\T(\bm{w}+\oP(1))
\notag \\
& \ \phantom{=} + \oP(\|\bm{\theta}-\bm{\theta}^*\|^2) + \oP\bigg(\frac{1}{\sqrt{n}}\|\bm{\theta}-\bm{\theta}^*\|\bigg).
\label{forthm3}
\end{align}
Since it follows from Theorem \ref{thm1} and \eqref{asymtau} that $\|\hat{\bm{\theta}}-\bm{\theta}^*\|=\oP(1)$ and $\|\hat{\bm{\tau}}-\bm{\tau}^*\|=\OP(\|\hat{\bm{\theta}}-\bm{\theta}^*\|)$, we obtain
\begin{align*}
& \frac{1}{n}\sum_{i=1}^ng(\hat{\bm{\xi}};x_i,y_i)
\notag \\
& = -\frac{1}{2}\bigg(\hat{\bm{\theta}}-\bm{\theta}^*-\frac{\bm{T}^{-1}\bm{w}+\oP(1)}{\sqrt{n}}\bigg)^{\T}(\bm{T}+\oP(1))\bigg(\hat{\bm{\theta}}-\bm{\theta}^*-\frac{\bm{T}^{-1}\bm{w}+\oP(1)}{\sqrt{n}}\bigg)
\notag \\
& \ \phantom{=} + \frac{\bm{w}^\T\bm{T}^{-1}\bm{w}}{2n}+\oP\bigg(\frac{1}{n}\bigg).
\end{align*}
Moreover, since it follows from $\hat{\bm{\xi}}=\argsup_{\bm{\xi}}\sum_{i=1}^ng(\bm{\xi};x_i,y_i)$ that
\begin{align*}
\bigg\|\hat{\bm{\theta}}-\bm{\theta}^*-\frac{\bm{T}^{-1}\bm{w}+\oP(1)}{\sqrt{n}}\bigg\|^2 = \oP\bigg(\frac{1}{n}\bigg),
\end{align*}
we obtain
\begin{align}
\sqrt{n}(\hat{\bm{\theta}}-\bm{\theta}^*)=\bm{T}^{-1}\bm{w}+\oP(1),
\label{asymtheta}
\end{align}
which yields the asymptotic distribution of $\hat{\bm{\theta}}$.
Also, we obtain from combining \eqref{asymtau} and \eqref{asymtheta} that 
\begin{align*}
\sqrt{n}(\hat{\bm{\tau}}-\bm{\tau}^*)=\bm{A}^{\T}\bm{T}^{-1}\bm{w}+\oP(1),
\end{align*}
which yields the asymptotic distribution of $\hat{\bm{\tau}}$.
\end{proof}

With the above preparation, let us consider
\begin{align*}
\text{AIC} = -2\sum_{i=1}^n\log f(x_i,y_i;\hat{\bm{\xi}}) + 2\E[c^{\rm limit}]
\end{align*}
according to the original definition of AIC.
Here, the second term represents the asymptotic bias when the Kullback-Leibler risk (exactly its double minus a constant) is estimated by $-2\sum_{i=1}^n\log f(x_i,y_i;\hat{\bm{\xi}})$.
In this paper, for simplicity, we use convergence in distribution, $\dlim_{n\to\infty}$, instead of convergence in mean, and evaluate the asymptotic bias by defining
\begin{align*}
c^{\rm limit} = \dlim_{n\to\infty}\bigg\{\sum_{i=1}^ng(\hat{\bm{\xi}};x_i,y_i)-\sum_{i=1}^ng(\hat{\bm{\xi}};\tilde{x}_i,\tilde{y}_i)\bigg\}.
\end{align*}
Here, $(\tilde{x}_i,\tilde{y}_i)$ is a copy of $(x_i,y_i)$, that is, a random variable that follows the same distribution independently of $(x_i,y_i)$.
From Theorem \ref{thm2}, the asymptotic bias is evaluated as follows.

\begin{thm}
Under Assumption \ref{asm1}, for the model in \eqref{model2}, where the regression function is continuous at the change-points, it holds that
\begin{align*}
\E[c^{\rm limit}] = p(m+1).
\end{align*}
\label{thm3}
\end{thm}

\begin{proof}
Letting $\tilde{\bm{w}}$ be a copy of $\bm{w}$, it follows from \eqref{forthm3} and \eqref{asymtheta} that
\begin{align*}
& \sum_{i=1}^ng(\hat{\bm{\xi}};\tilde{x_i},\tilde{y_i})
\\
& = n\bigg\{-\frac{1}{2}(\hat{\bm{\theta}}-\bm{\theta}^*)^\T(\bm{T}+\oP(1))(\hat{\bm{\theta}}-\bm{\theta}^*)+\frac{1}{\sqrt{n}}(\hat{\bm{\theta}}-\bm{\theta}^*)^\T(\tilde{\bm{w}}+\oP(1))\bigg\}
\\
& = -\frac{1}{2}(\bm{T}^{-1}\bm{w}+\oP(1))^\T(\bm{T}+\oP(1))(\bm{T}^{-1}\bm{w}+\oP(1))+(\bm{T}^{-1}\bm{w}+\oP(1))^\T(\tilde{\bm{w}}+\oP(1))
\\
& = -\frac{1}{2}\bm{w}^\T\bm{T}^{-1}\bm{w}+\bm{w}^\T\bm{T}^{-1}\tilde{\bm{w}}+\oP(1).
\end{align*}
Similarly we obtain
\begin{align*}
\sum_{i=1}^ng(\hat{\bm{\xi}};x_i,y_i) = -\frac{1}{2}\bm{w}^\T\bm{T}^{-1}\bm{w}+\bm{w}^\T\bm{T}^{-1}\bm{w}+\oP(1) = \frac{1}{2}\bm{w}^\T\bm{T}^{-1}\bm{w}+\oP(1),
\end{align*}
and these combination implies
\begin{align*}
2\bigg(\sum_{i=1}^ng(\hat{\bm{\xi}};x_i,y_i)-\sum_{i=1}^ng(\hat{\bm{\xi}};\tilde{x}_i,\tilde{y}_i)\bigg) = 2(\bm{w}-\tilde{\bm{w}})^{\T}\bm{T}^{-1}\bm{w}+\oP(1). 
\end{align*}
Since $\E[(\bm{w}-\tilde{\bm{w}})^\T\bm{T}^{-1}\bm{w}]=p(m+1)$, the proof is complete.
\end{proof}

\noindent
From this result, we propose to use
\begin{align}
\text{AIC} = -2\sum_{i=1}^n\log f(x_i,y_i;\hat{\bm{\xi}}) + 2p(m+1)
\label{aic1}
\end{align}
as AIC for the segmented regression model when the regression function is continuous at the change-points.
This means that the penalties do not change depending on the parameters, that is, the penalty for the change-point parameter is also $2$.

\section{AIC for discontinuous segmented regression models}\label{sec4}

In this section, we consider the case where the regression function $\sum_{k=1}^{m+1}I_{(\tau^{[k-1]},\tau^{[k]}]}(x)\allowbreak\phi(x;\bm{\theta}^{[k]})$ is discontinuous at the change-points $\tau^{[1]},\tau^{[2]},\ldots,\tau^{[m]}$.
Let us denote the maximum likelihood estimator of $\bm{\theta}$ when the change-point is fixed at $\bm{\tau}$ as $\hat{\bm{\theta}}_{\bm{\tau}}=(\hat{\bm{\theta}}_{\bm{\tau}}^{[1]\T},\hat{\bm{\theta}}_{\bm{\tau}}^{[2]\T},\ldots,\allowbreak\hat{\bm{\theta}}_{\bm{\tau}}^{[m+1]\T})^\T$, i.e. $\hat{\bm{\theta}}_{\bm{\tau}}=\argsup_{\bm{\theta}}(\sum_{i=1}^n g(\bm{\theta},\bm{\tau};x_i,y_i))$, and $g(\hat{\bm{\theta}}_{\bm{\tau}},\bm{\tau};x_i,y_i)$ as $\hat{g}(\bm{\tau};x_i,y_i)$. 
When $x_i$ is determined to be in $(\tau^{[k-1]},\tau^{[k]}]$, $g(\bm{\theta},\bm{\tau};x_i,y_i)$ depends only on $\bm{\theta}^{[k]}$; in this case it is written as $g(\bm{\theta}^{[k]};x_i,y_i)$.
Since the regression function is not assumed to be continuous at the change-points, $\hat{\bm{\theta}}_{\bm{\theta}}^{[k]}$ is $\bm{\theta}^{[k]}$ that maximizes $\sum_{i:\tau^{[k-1]}<x_i\le\tau^{[k]}}g(\bm{\theta}^{[k]};x_i,y_i)$. 
This is the maximum likelihood estimator when a generalized linear regression model is fitted to a set of $(x_i,y_i)$ with $\tau^{[k-1]}<x_i\le\tau^{[k]}$, and therefore $\hat{\bm{\theta}}_{\bm{\tau}}^{[k]}$ at $\bm{\tau}=\bm{\tau}^*$, i.e. $\hat{\bm{\theta}}_{\bm{\tau}^*}^{[k]}$, is consistent.
In this notation, the weak limit considered in the derivation of AIC is
\begin{align*}
c^{\rm limit} = \dlim_{n\to\infty}\bigg(\sup_{\bm{\tau}}\sum_{i=1}^n\hat{g}(\bm{\tau};x_i,y_i)-\sum_{i^{\dagger}=1}^n\hat{g}\bigg(\argsup_{\bm{\tau}}\sum_{i=1}^n\hat{g}(\bm{\tau};x_i,y_i);\tilde{x}_{i^{\dagger}},\tilde{y}_{i^{\dagger}}\bigg)\bigg).
\end{align*}
Here, we take $\sup_{\bm{\tau}}$ and $\argsup_{\bm{\tau}}$ in the set of $\bm{\tau}$ such that $\sum_{i=1}^n\hat{g}(\bm{\tau};x_i,y_i)$ is $\OP(1)$ or positive.

To investigate the asymptotic properties of $\hat{\bm{\tau}}=\argsup_{\bm{\tau}}\sum_{i=1}^n\hat{g}(\bm{\tau};x_i,y_i)$, we separately consider the log-likelihood for the cases $\bm{\tau}-\bm{\tau}^*=\O(n^{-1})$ and $\bm{\tau}-\bm{\tau}^*\neq\O(n^{-1})$ and $\bm{\tau}-\bm{\tau}^*=\O(n^{-1})$.
Here, $\bm{\tau}-\bm{\tau}^*=\O(n^{-1})$ means that $\tau^{[k]}-\tau^{[k]*}=\O(n^{-1})$ for all $k$, and $\bm{\tau}-\bm{\tau}^*\neq\O(n^{-1})$ means that $\tau^{[k]}-\tau^{[k]*}\neq\O(n^{-1})$ for some $k$.
When $\bm{\tau}-\bm{\tau}^*=\O(n^{-1})$, the difference between $\sum_{i:\tau^{[k-1]}<x_i\le\tau^{[k]}}g(\bm{\theta}^{[k]};x_i,y_i)$ and $\sum_{i:\tau^{[k-1]*}<x_i\le\tau^{[k]*}}g(\bm{\theta}^{[k]};x_i,y_i)$ is small, $\hat{\bm{\theta}}_{\bm{\tau}}^{[k]}$ is obviously consistent, as well as $\hat{\bm{\theta}}_{\bm{\tau}^*}^{[k]}$, i.e. $\hat{\bm{\theta}}_{\bm{\tau}}^{[k]}\xrightarrow{\rm p}\bm{\theta}^{[k]*}$. 
Then, taking the Taylor expansion of $\sum_{i:\tau^{[k-1]}<x_i\le\tau^{[k]}}\allowbreak g'(\hat{\bm{\theta}}_{\bm{\tau}}^{[k]};x_i,y_i)=\bm{0}_p$ around $\hat{\bm{\theta}}_{\bm{\tau}}^{[k]}=\bm{\theta}^{[k]*}$, we obtain
\begin{align*}
\bm{0}_p=\sum_{i:\tau^{[k-1]}<x_i\le\tau^{[k]}}g'(\bm{\theta}^{[k]*};x_i,y_i)+\bigg(\sum_{i:\tau^{[k-1]}<x_i\le\tau^{[k]}}g''(\bm{\theta}^{[k]*};x_i,y_i)\bigg)(\hat{\bm{\theta}}_{\bm{\tau}}^{[k]}-\bm{\theta}^{[k]*})(1+\oP(1)),
\end{align*}
where where $g'(\bm{\theta};x_i,y_i)$ and $g''(\bm{\theta};x_i,y_i)$ are the first-order derivative vector and second-order derivative matrix of $g(\bm{\theta};x_i,y_i)$ with respect to $\bm{\theta}$.
Therefore, it can also be written that
\begin{align*}
\hat{\bm{\theta}}_{\bm{\tau}}^{[k]}=\bm{\theta}^{[k]*}-\bigg(\sum_{i:\tau^{[k-1]}<x_i\le\tau^{[k]}}g''(\bm{\theta}^{[k]*};x_i,y_i)\bigg)^{-1}\bigg(\sum_{i:\tau^{[k-1]}<x_i\le\tau^{[k]}}g'(\bm{\theta}^{[k]*};x_i,y_i)\bigg)(1+\oP(1)),
\end{align*}
and we obtain
\begin{align}
& \hat{\bm{\theta}}_{\bm{\tau}}^{[k]}-\hat{\bm{\theta}}_{\bm{\tau}^*}^{[k]}
\notag \\
& = \bigg(\sum_{i:\tau^{[k-1]*}<x_i\le\tau^{[k]*}}g''(\bm{\theta}^{[k]*};x_i,y_i)\bigg)^{-1}\bigg(\sum_{i:\tau^{[k-1]*}<x_i\le\tau^{[k]*}}g'(\bm{\theta}^{[k]*};x_i,y_i)\bigg)(1+\oP(1))
\notag \\
& \ \phantom{=} -\bigg(\sum_{i:\tau^{[k-1]}<x_i\le\tau^{[k]}}g''(\bm{\theta}^{[k]*};x_i,y_i)\bigg)^{-1}\bigg(\sum_{i:\tau^{[k-1]}<x_i\le\tau^{[k]}}g'(\bm{\theta}^{[k]*};x_i,y_i)\bigg)(1+\oP(1))
\notag \\
& = \bigg(\sum_{i:\tau^{[k-1]*}<x_i\le\tau^{[k]*}}g''(\bm{\theta}^{[k]*};x_i,y_i)\bigg)^{-1}\OP(1)(1+\oP(1)) = \OP(n^{-1}).
\label{keyasym}
\end{align}
In addition, since we know $\hat{\bm{\theta}}_{\bm{\tau}^*}^{[k]}-\bm{\theta}^{[k]*}=\OP(n^{-1/2})$ as a fundamental property of the maximum likelihood estimator for generalized linear models, it holds $\hat{\bm{\theta}}_{\bm{\tau}}^{[k]}-\bm{\theta}^{[k]*}=\OP(n^{-1/2})$. 
Hence, by taking the Taylor expansion of the maximum log-likelihood with $\bm{\tau}$ fixed and applying \eqref{keyasym}, we obtain
\begin{align*}
& \sum_{i=1}^ng(\hat{\bm{\theta}}_{\bm{\tau}},\bm{\tau};x_i,y_i)-\sum_{i=1}^ng(\bm{\theta}^*,\bm{\tau};x_i,y_i)
\\
& = -\frac{1}{2}\sum_{k=1}^{m+1}\sum_{i:\tau^{[k-1]}<x_i\le\tau^{[k]}}(\bm{\theta}^{[k]*}-\hat{\bm{\theta}}_{\bm{\tau}}^{[k]})^{\T}g''(\bm{\theta}^{[k]*};x_i,y_i)(\bm{\theta}^{[k]*}-\hat{\bm{\theta}}_{\bm{\tau}}^{[k]})+\oP(1)
\\
& = -\frac{1}{2}\sum_{k=1}^{m+1}\sum_{i:\tau^{[k-1]*}<x_i\le\tau^{[k]*}}(\bm{\theta}^{[k]*}-\hat{\bm{\theta}}_{\bm{\tau}^*}^{[k]})^{\T}g''(\bm{\theta}^{[k]*};x_i,y_i)(\bm{\theta}^{[k]*}-\hat{\bm{\theta}}_{\bm{\tau}^*}^{[k]})+\oP(1)
\\
& = \sum_{i=1}^ng(\hat{\bm{\theta}}_{\bm{\tau}^*},\bm{\tau}^*;x_i,y_i)+\oP(1).
\end{align*}
Then, defining a two-sided random walk with negative drift as
\begin{align*}
Q^{[k]}(\tau^{[k]})
& \equiv I_{\{\tau^{[k]}<\tau^{[k]*}\}}\sum_{i:\tau^{[k]}<x_i\le\tau^{[k]*}}(g(\bm{\theta}^{[k+1]*};x_i,y_i)-g(\bm{\theta}^{[k]*};x_i,y_i))
\\
& \ \phantom{\equiv} + I_{\{\tau^{[k]*}<\tau^{[k]}\}}\sum_{i:\tau^{[k]*}<x_i\le\tau^{[k]}}(g(\bm{\theta}^{[k]*};x_i,y_i)-g(\bm{\theta}^{[k+1]*};x_i,y_i)),
\end{align*}
it can be seen that 
\begin{align}
& \sum_{i=1}^ng(\hat{\bm{\theta}}_{\bm{\tau}},\bm{\tau};x_i,y_i)-\sum_{i=1}^ng(\hat{\bm{\theta}}_{\bm{\tau}^*},\bm{\tau}^*;x_i,y_i)
\notag \\
& = \sum_{i=1}^ng(\bm{\theta}^*,\bm{\tau};x_i,y_i)+\oP(1) = \sum_{k=1}^mQ^{[k]}(\tau^{[k]})+\oP(1) = \OP(1).
\label{order1}
\end{align}
Using asymptotics for regular statistical models, it can also be seen that
\begin{align}
& \sum_{i=1}^ng(\hat{\bm{\theta}}_{\bm{\tau}^*},\bm{\tau}^*;x_i,y_i)
\notag \\
& = \frac{1}{2}\sum_{k=1}^{m+1}\bigg(\sum_{i:\tau^{[k-1]*}<x_i\le\tau^{[k]*}}g'(\bm{\theta}^{[k]*};x_i,y_i)\bigg)^{\T}\bigg(-\sum_{i:\tau^{[k-1]*}<x_i\le\tau^{[k]*}}g''(\bm{\theta}^{[k]*};x_i,y_i)\bigg)^{-1}
\notag \\
& \ \phantom{= \frac{1}{2}\sum_{k=1}^{m+1}} \times \bigg(\sum_{i:\tau^{[k-1]*}<x_i\le\tau^{[k]*}}g'(\bm{\theta}^{[k]*};x_i,y_i)\bigg)+\oP(1)
\notag \\
& = \frac{1}{2}\sum_{k=1}^{m+1}\bm{u}^{[k]\T}\bm{u}^{[k]}+\oP(1)=\OP(1),
\label{order2}
\end{align}
where $\bm{u}^{[k]}$ is a random vector following a multi-dimensional Gaussian distribution $\N(\bm{0}_p,\allowbreak\bm{I}_{p\times p})$ independently.
From \eqref{order1} and \eqref{order2}, when $\bm{\tau}-\bm{\tau}^*=\O(n^{-1})$, we have $\sum_{i=1}^ng(\hat{\bm{\theta}}_{\bm{\tau}},\bm{\tau};x_i,\allowbreak y_i)=\OP(1)$.
On the other hand, when $\bm{\tau}-\bm{\tau}^*\neq\O(n^{-1})$, we can show that 
\begin{align}
\P\bigg(\sum_{i=1}^ng(\hat{\bm{\theta}}_{\bm{\tau}},\bm{\tau};x_i,y_i)>-M\bigg)\to 0 
\label{gtoinf}
\end{align}
for any $M>0$ (see Appendix).
Thus, we take $\sup_{\bm{\tau}}$ and $\argsup_{\bm{\tau}}$ in the set of $\bm{\tau}$ such that $\bm{\tau}-\bm{\tau}^*=\O(n^{-1})$.
From \eqref{order1} and \eqref{order2}, it holds
\begin{align}
\sup_{\bm{\tau}}\sum_{i=1}^ng(\hat{\bm{\theta}}_{\bm{\tau}},\bm{\tau};x_i,y_i) = \sum_{k=1}^m\sup_{\tau^{[k]}}Q^{[k]}(\tau^{[k]})+\frac{1}{2}\sum_{k=1}^{m+1}\bm{u}^{[k]\T}\bm{u}^{[k]}+\oP(1),
\label{supres1}
\end{align}
and furthermore,
\begin{align}
\argsup_{\bm{\tau}}\sum_{i=1}^n\hat{g}(\bm{\tau};x_i,y_i) = \check{\bm{\tau}}+\oP(1),
\label{supres2}
\end{align}
where $\check{\tau}^{[k]}=\argsup_{\tau^{[k]}}Q^{[k]}(\tau^{[k]})$ and $\check{\bm{\tau}}=(\check{\tau}^{[1]},\check{\tau}^{[2]},\ldots,\check{\tau}^{[m]})^{\T}$.
Then, the asymptotic bias of the Kullback-Leibler risk can be expressed using the two-sided random walk as follows.
Although the proof is similar to that of \cite{Nin15}, the definition of the asymptotic bias is given by not using a copy of $\hat{\bm{\xi}}$, which reduces the complexity of the expression (see Appendix).

\begin{thm}
Under Assumption \ref{asm1}, for the model in \eqref{model2}, where the regression function is discontinuous at the change-points, it holds that
\begin{align*}
\E[c^{\rm limit}] = \sum_{k=1}^m\E\bigg[\sup_{\tau}Q^{[k]}(\tau)-\tilde{Q}^{[k]}\bigg(\argsup_{\tau}Q^{[k]}(\tau)\bigg)\bigg] + p(m+1),
\end{align*}
where $\tilde{Q}^{[k]}(\cdot)$ is a copy of $Q^{[k]}(\cdot)$.
\label{thm4}
\end{thm}

\noindent
From this theorem, $\E[\sup_{\tau}Q^{[k]}(\tau)-\tilde{Q}^{[k]}(\argsup_{\tau}Q^{[k]}(\tau))]$ can be regarded as the asymptotic bias due to the change-point parameter $\tau^{[k]}$ and $p$ the asymptotic bias due to the non-change-point parameter $\bm{\theta}^{[k]}$.

Up to this point, we have treated asymptotics with a fixed amount of change $\bm{\theta}^{[k+1]*}-\bm{\theta}^{[k]*}$; however, the following is often assumed in change-point analysis (see, for example, \citealt{CsoH97}).

\begin{asm}
Letting $\alpha_n$ be some constant sequence satisfying $\O(1)\neq\alpha_n=\o(n)$ and $\bm{\Delta}_{\bm{\theta}^*}^{[k]}$ be some $p$-dimensional constant vector, it holds
\begin{align*}
\bm{\theta}^{[k+1]*}-\bm{\theta}^{[k]*} = \alpha_n^{-1/2}\bm{\Delta}_{\bm{\theta}^*}^{[k]} \qquad (k\in\{1,2,\ldots,m\}).
\end{align*}
\label{asm2}
\end{asm}

\noindent
In ordinary statistical analysis, an assumption like Assumption \ref{asm2} on the true values of different parameters does not change the asymptotic properties of estimators; in change-point analysis, however, it will change the asymptotic properties.
That is, whether or not to set Assumption \ref{asm2} is essential, and we will assume it in the following.
This is because setting it makes the limiting structure ambiguous as to whether or not there is change, and this is a more realistic asymptotics.
Since the actual data size is finite, the asymptotics adapted to the situation in which it is not completely clear whether or not there is a change in the true structure would clearly be more realistic.

Assumption \ref{asm2} is not only natural, but also allows us to explicitly express the asymptotic bias of the Kullback-Leibler risk as follows.
Although the proof method is similar to that in \cite{Nin15}, we provide it since the definition of asymptotic bias is given in a different form (see Appendix).

\begin{thm}
Under Assumptions \ref{asm1} and \ref{asm2}, for the model in \eqref{model2}, where the regression function is discontinuous at the change-points, it holds that
\begin{align*}
\E[c^{\rm limit}] = 3m + p(m+1).
\end{align*}
\label{thm5}
\end{thm}

\noindent
From this result, we propose to use
\begin{align}
\text{AIC} = -2\log f(\bm{x},\bm{y};\hat{\bm{\xi}}) + 6m + 2p(m+1)
\label{aic2}
\end{align}
as AIC for the segmented regression model when the regression function is discontinuous at the change-points.
This means that the penalty for the change-point parameter is $6$, while the penalty for the non-change-point parameter is still $2$.

\section{Numerical experiments}\label{sec5}

We examine the performance of AIC given by \eqref{aic1} and \eqref{aic2} through numerical experiments using the linear segmented regression model in \eqref{model}, where $\{x_i\}_{1\le i\le n}$ and $\{\epsilon_i\}_{1\le i\le n}$ are independent random variables following the uniform distribution $\U(0,1)$ on the interval $[0,1]$ and the standard Gaussian distribution $\N(0,1)$, respectively.
As one of the comparisons, we consider a criterion that incorrectly captures the non-regularity of the change-point parameter, and denote it as AIC$_{\rm naive}$.
Specifically, we set the penalty term of AIC$_{\rm naive}$ to $8m+4$ by assigning a value of 6 to the change-point even in continuous cases and to $6m+4$ by assigning a value of 2 to the change-point even in discontinuous cases.
In addition, as another comparison, we consider BIC, where the penalty term is set to $(2m+2)\log n$ in continuous cases and to $(3m+2)\log n$ in discontinuous cases.

First, we check whether the asymptotic bias evaluation, which is the basis of our AIC penalty term, can approximate the bias of the maximum log-likelihood.
Table \ref{tab1} provides the true biases obtained from the Monte Carlo method for various change-point numbers $m$, sample sizes $n$, and true values of the regression parameter vector $\bm{\theta}^*$.
The number of simulation repetitions is set to 1,000 throughout the paper.
The true value of the change-point parameter vector $\bm{\tau}$ is set to $\bm{\tau}^*=0.5\ (m=1),\ (0.3,0.7)^\T\ (m=2),\ (0.25,0.5,0.75)^\T\ (m=3),\ (0.2,0.4,0.6,0.8)^\T\ (m=4)$, and the vector $\bm{\theta}_0$ used to provide the true value of the parameter $\bm{\theta}$ is set to $(-0.5,0,0.5,0)^\T\ (m=1),\ (-0.5,0,0.5,0,\allowbreak -0.5,0)^\T\ (m=2),\ (-0.5,0,0.5,0,-0.5,0,0.5,0)^\T\ (m=3),\ (-0.5,0,0.5,0,-0.5,0,0.5,0,\allowbreak -0.5,0)^\T\ (m=4)$ for the continuous case and $(-1,0,1,-1)^\T\ (m=1),\ (-1,0,1,-0.6,-1,\allowbreak 0.8)^\T\ (m=2),\ (-1,0,1,-0.5,-1,0.5,1,-1)^\T\ (m=3),\ (-1,0,1,-0.4,-1,0.4,1,-0.8,\allowbreak -1,0.8)^\T\ (m=4)$ for the discontinuous case.
From these values, we can see that our AIC penalty term accurately approximates the bias in the continuous case.
In the discontinuous case, since the model is too complex, our AIC penalty term only maintains a certain degree of approximation accuracy; however, it is clearly closer to bias than the AIC$_{\rm naive}$ and BIC penalty terms.

\begin{table}[!t]
\caption{Bias evaluation of maximum log-likelihood: For various change-point numbers $m$, sample sizes $n$, and true values of parameter vector $\bm{\theta}^*$ defined using appropriate $\bm{\theta}_0$, each value is obtained using the Monte Carlo method.}
\begin{center}
\begin{tabular}{cccccccccc}
\hline
 & & \multicolumn{4}{c}{continuous case} & \multicolumn{4}{c}{discontinuous case} \\
$n$ & $\bm{\theta}^*$ & $m=1$ & $m=2$ & $m=3$ & $m=4$ & $m=1$ & $m=2$ & $m=3$ & $m=4$ \\ 
\hline
\multirow{4}{*}{100} & $\bm{\theta}_0$ & 4.08 & 5.87 & 7.51 & 8.74 & 7.83 & 14.36 & 19.98 & 25.12 \\
 & $2\bm{\theta}_0$ & 4.27 & 6.14 & 7.60 & 8.75 & 7.76 & 14.36 & 20.29 & 25.83 \\ 
 & $3\bm{\theta}_0$ & 4.36 & 6.19 & 7.77 & 8.90 & 7.42 & 14.11 & 21.46 & 26.84 \\ 
 & $4\bm{\theta}_0$ & 4.32 & 6.26 & 7.90 & 9.05 & 7.36 & 15.13 & 22.90 & 29.10 \\
\hline
\multirow{4}{*}{200} & $\bm{\theta}_0$ & 4.14 & 6.05 & 7.76 & 9.29 & 8.21 & 14.89 & 21.03 & 27.05 \\
 & $2\bm{\theta}_0$ & 4.39 & 6.11 & 7.84 & 9.42 & 7.52 & 14.19 & 20.76 & 27.28 \\
 & $3\bm{\theta}_0$ & 4.32 & 6.08 & 7.97 & 9.40 & 7.16 & 13.96 & 20.76 & 27.67 \\
 & $4\bm{\theta}_0$ & 4.21 & 6.27 & 8.06 & 9.66 & 7.56 & 14.53 & 21.74 & 29.52 \\
\hline
\multirow{4}{*}{300} & $\bm{\theta}_0$ & 4.26 & 6.01 & 7.89 & 9.60 & 8.28 & 14.97 & 21.58 & 28.02 \\
 & $2\bm{\theta}_0$ & 4.45 & 6.10 & 7.92 & 9.64 & 7.13 & 13.94 & 20.61 & 26.75 \\
 & $3\bm{\theta}_0$ & 4.30 & 5.96 & 8.07 & 9.67 & 6.84 & 13.70 & 20.21 & 27.36 \\
 & $4\bm{\theta}_0$ & 4.21 & 6.20 & 8.33 & 9.69 & 7.14 & 14.35 & 20.96 & 29.00 \\
\hline
\end{tabular}
\end{center}
\label{tab1}
\end{table}

Next, to assess the performance of our AIC from the perspective of model selection, we evaluate the Kullback-Leibler divergence between the true and estimated distributions using the Monte Carlo method as the main index, and also provide the selection rate of change-points as a reference index.
The results for cases where the regression function is continuous and the number of true change-points is 1 and 2 are summarized in Table \ref{tab2} and Table \ref{tab4}, respectively, and the results for cases where the regression function is discontinuous and the number of true change-points is 1 and 2 are summarized in Table \ref{tab3} and Table \ref{tab5}, respectively.
Regarding Table \ref{tab2}, when the sample size and change are small, as in the top row, it is better to select a model with 0 change-points, and AIC selects the true model too often, returning results that are slightly inferior from the perspective of the main index.
In addition, as in the bottom row, where both the sample size and the change are large, and the model with 0 change-points is clearly inappropriate, AIC returns inferior results because it selects a slightly redundant model.
Otherwise, AIC clearly outperforms AIC$_{\rm naive}$ and BIC.
Regarding Table \ref{tab3}, when the change is small, or when both the sample size and the change are large, AIC returns results that are equivalent to or slightly inferior to BIC.
In other cases, although the difference is not as large as in Table \ref{tab2}, AIC outperforms BIC.
AIC$_{\rm naive}$ is always clearly inferior.
In Tables \ref{tab4} and \ref{tab5}, there are large and small changes, making it difficult to grasp the properties of each criterion.
However, we can say that while AIC$_{\rm naive}$ and BIC in continuous cases and BIC in discontinuous cases are sometimes slightly superior to AIC, the result that AIC basically performs better remains unchanged.

\begin{table}[!p]
\caption{Evaluation of Kullback-Leibler divergence and selection rate of change-points when the regression function is continuous and the number of true change-points is 1: Each value is obtained using the Monte Carlo method for various sample sizes $n$ and true values of the parameter vector $\bm{\theta}^*$.}
\begin{center}
\renewcommand{\baselinestretch}{1.1}\selectfont
\begin{tabular}{ccccccccc} \hline
 & & & \multicolumn{2}{c}{K-L ($\times 100$)} & \multicolumn{4}{c}{selection rate (\%)} \\
 $n$ & $\bm{\theta}^*$ & criterion & mean & med &  0 & 1 & 2 & 3$^+$ \\ \hline
\multirow{12}{*}{100} & & AIC & 
5.61	& 	4.52	&	42.5	&	49.1	& \ 6.3 &	2.1 \\
 & $(0,0,2.5,-1.25)$ & AIC$_{\rm naive}$ &
5.52	&	4.70 &	86.7	&	13.3	&	\ 0.0 &	0.0 \\
 & & BIC &
5.52 &	4.71 &	83.6	&	16.3	&	\ 0.1 &	0.0 \\ \cline{2-9}
 & & AIC &
5.94 &	5.17 &	29.9	&	60.2	&	\ 7.5 &	2.4 \\
 & $(0,0,3,-1.5)$ & AIC$_{\rm naive}$ &
6.62 &	6.00 &	77.8	&	22.2	&	\ 0.0 &	0.0 \\
 & & BIC &
6.52 &	5.95 &	73.9	&	26.0	&	\ 0.1 &	0.0 \\ \cline{2-9}
 & & AIC &
6.08 &	5.42 &	20.2	&	69.1	&	\ 7.9 &	2.8 \\
 & $(0,0,3.5,-1.75)$ & AIC$_{\rm naive}$ &
7.48 &	7.32 &	67.0	&	33.0	&	\ 0.0 &	0.0 \\
 & & BIC &
7.25 &	7.21 &	61.8	&	38.0	&	\ 0.2 &	0.0 \\ \cline{2-9}
 & & AIC &
5.95 &	4.63 &	11.9	&	78.0	&	\ 7.1 &	3.0 \\
 & $(0,0,4,-2)$ & AIC$_{\rm naive}$ &
7.88 &	8.75 &	52.5	&	47.4	&	\ 0.1 &	0.0 \\
 & & BIC &
7.63 &	8.60 &	48.6	&	51.2	&	\ 0.2 &	0.0 \\ \hline
\multirow{12}{*}{200} & & AIC & 
3.04	& 2.67	& 17.7	& 68.1	& 10.4	& 3.8 \\
 & $(0,0,2.5,-1.25)$ & AIC$_{\rm naive}$ &
3.71	& 3.69	& 63.3	& 36.6	& \ 0.1	& 0.0 \\
 & & BIC &
3.79	& 3.72	& 66.8	& 33.1	& \ 0.1	& 0.0 \\ \cline{2-9}
 & & AIC & 
3.08	& 2.38	& \ 9.0	& 76.2	& 10.9	& 3.9 \\
 & $(0,0,3,-1.5)$ &	AIC$_{\rm naive}$ &
3.91	& 4.70	& 45.2	& 54.5	& \ 0.3	& 0.0 \\
 & & BIC &
4.06	& 4.78	& 48.9	& 50.8	& \ 0.3	& 0.0 \\ \cline{2-9}
 & & AIC & 
2.95	& 2.24	& \ 3.3	& 82.4	& 10.6	& 3.7 \\
 & $(0,0,3.5,-1.75)$ & AIC$_{\rm naive}$ &
3.68	& 2.61	& 27.1	& 72.7	& \ 0.2	& 0.0 \\
 & & BIC &
3.84	& 2.85	& 30.2	& 69.6	& \ 0.2	& 0.0 \\ \cline{2-9}
 & & AIC & 
2.94	& 2.22	& \ 1.4	& 83.4	& 11.4	& 3.8 \\
 & $(0,0,4,-2)$ & AIC$_{\rm naive}$ &
3.25	& 2.12	& 14.0	& 85.7	& \ 0.3	& 0.0 \\
 & & BIC &
3.38	& 2.18	& 15.8	& 83.9	& \ 0.3	& 0.0 \\ \hline
\multirow{12}{*}{300}	 & & AIC & 
2.02	& 1.63	& \ 8.3	& 77.1	& \ 9.5	& 5.1 \\
 & $(0,0,2.5,-1.25)$ & AIC$_{\rm naive}$ &
2.61	& 2.89	& 44.5	& 55.2	& \ 0.3	& 0.0 \\
 & & BIC &
2.82	& 3.32	& 52.3	& 47.4	& \ 0.3	& 0.0 \\ \cline{2-9}
 & & AIC & 
2.02	& 1.51	& \ 2.7	& 81.3	& \ 9.8	& 6.2 \\
 & $(0,0,3,-1.5)$ &	AIC$_{\rm naive}$ &
2.34	& 1.66	& 22.3	& 77.3	& \ 0.4	& 0.0 \\
 & & BIC &
2.62	& 1.90	& 29.4	& 70.3	& \ 0.3	& 0.0 \\ \cline{2-9}
 & & AIC & 
1.99	& 1.53	& \ 0.6	& 83.4	& 10.3	& 5.7 \\
 & $(0,0,3.5,-1.75)$ & AIC$_{\rm naive}$ &
2.00	& 1.32	& \ 9.2	& 90.4	& \ 0.4	& 0.0 \\
 & & BIC &
2.22	& 1.39	& 13.1	& 86.5	& \ 0.4	& 0.0 \\ \cline{2-9}
 & & AIC & 
1.98	& 1.50	& \ 0.0	& 83.4	& 11.4	& 5.2 \\
 & $(0,0,4,-2)$ & AIC$_{\rm naive}$ &
1.74	& 1.24	& \ 3.0	& 96.6	& \ 0.4	& 0.0 \\
 & & BIC &
1.84	& 1.24	& \ 4.5	& 95.2	& \ 0.3	& 0.0 \\ \hline
\end{tabular}
\end{center}
\label{tab2}
\end{table}

\begin{table}[!p]
\caption{Evaluation of Kullback-Leibler divergence and selection rate of change-points when the regression function is discontinuous and the number of true change-points is 1: Each value is obtained using the Monte Carlo method for various sample sizes $n$ and true values of the parameter vector $\bm{\theta}^*$.}
\begin{center}
\renewcommand{\baselinestretch}{1.1}\selectfont
\begin{tabular}{ccccccccc} \hline
 & & & \multicolumn{2}{c}{K-L ($\times 10$)} & \multicolumn{4}{c}{selection rate (\%)} \\
 $n$ & $\bm{\theta}^*$ & criterion & mean & med &  0 & 1 & 2 & 3$^+$ \\ \hline
\multirow{12}{*}{100}	 & & AIC & 
		1.23 &	1.09 &	34.2 &	47.2 &	15.5 &	\ 3.1 \\
 & $(0,0,1.2,0)$ & AIC$_{\rm naive}$ &
		1.56 &	1.41 &	\ 7.3 &	26.2 &	35.9 &	30.6 \\
 & & BIC &
		1.13 &	1.05 &	77.7 &	21.0 &	\ 1.3 &	\ 0.0 \\ \cline{2-9}
 & & AIC & 
		1.37 &	1.19 &	\ 6.3 &	70.3 &	19.4 &	\ 4.0 \\
 & $(0,0,1.8,0)$ & AIC$_{\rm naive}$ &
		1.75 &	1.65 &	\ 0.6 &	28.6 &	36.6 &	34.2 \\
 & & BIC &
		1.46 &	1.33 &	30.5 &	67.7 &	\ 1.8 &	\ 0.0 \\ \cline{2-9}
 & & AIC & 
		1.14 &	1.11 &	22.0 &	61.3 &	13.9 &	\ 2.8 \\
 & $(0,0,-1.4,4)$ & AIC$_{\rm naive}$ &
		1.41 &	1.29 &	\ 2.6 &	33.1 &	37.9 &	26.4 \\
 & & BIC &
		1.19 &	1.18 &	61.7 &	37.3 &	\ 0.9 &	\ 0.1 \\ \cline{2-9}
 & & AIC & 
		1.10 &	0.96 &	\ 2.1 &	76.9 &	18.1 &	\ 2.9 \\
 & $(0,0,-2.4,6)$ & AIC$_{\rm naive}$ &
		1.42 &	1.29 &	\ 0.1 &	32.7 &	38.6 &	28.6 \\
 & & BIC &
		1.22 &	1.00 &	19.4 &	79.2 &	\ 1.3 &	\ 0.1 \\ \hline
\multirow{12}{*}{200}	 & & AIC & 
		0.60 &	0.60 &	21.1 &	60.0 &	17.0 &	\ 1.9 \\
 & $(0,0,1,0)$ & AIC$_{\rm naive}$ &
		0.79 &	0.74 &	\ 2.5 &	30.9 &	35.0 &	31.7 \\
 & & BIC &
		0.66 &	0.67 &	71.6 &	28.3 &	\ 0.1 &	\ 0.0 \\ \cline{2-9}
 & & AIC & 
		0.61 &	0.49 &	\ 0.7 &	76.7 &	20.2 &	\ 2.4 \\
 & $(0,0,1.6,0)$ & AIC$_{\rm naive}$ &
		0.84 &	0.77 &	\ 0.0 &	32.3 &	35.9 &	31.8 \\
 & & BIC &
		0.60 &	0.45 &	\ 9.5 &	89.9 &	\ 0.6 &	\ 0.0 \\ \cline{2-9}
 & & AIC & 
		0.59 &	0.52 &	12.6 &	71.1 &	14.3 &	\ 2.0 \\
 & $(0,0,-0.9,3)$ & AIC$_{\rm naive}$ &
		0.75 &	0.71 &	\ 1.8 &	35.7 &	36.3 &	26.2 \\
 & & BIC &
		0.69 &	0.73 &	59.0 &	40.7 &	\ 0.3 &	\ 0.0 \\ \cline{2-9}
 & & AIC & 
		0.56 &	0.71 &	\ 0.3 &	79.9 &	18.0 &	\ 1.8 \\
 & $(0,0,-1.9,5)$ & AIC$_{\rm naive}$ &
		0.75 &	0.71 &	\ 0.0 &	33.3 &	38.7 &	28.0 \\
 & & BIC &
		0.56 &	0.41 &	\ 9.2 &	90.2 &	\ 0.6 &	\ 0.0 \\ \hline
\multirow{12}{*}{300}	 & & AIC & 
		1.23 &	1.09 &	34.2 &	47.2 &	15.5 &	\ 3.1 \\
 & $(0,0,1.2,0)$ & AIC$_{\rm naive}$ &
		1.56 &	1.41 &	\ 7.3 &	26.2 &	35.9 &	30.6 \\
 & & BIC &
		1.13 &	1.05 &	77.7 &	21.0 &	\ 1.3 &	\ 0.0 \\ \cline{2-9}
 & & AIC & 
		1.37 &	1.19 &	\ 6.3 &	70.3 &	19.4 &	\ 4.0 \\
 & $(0,0,1.8,0)$ & AIC$_{\rm naive}$ &
		1.75 &	1.65 &	\ 0.6 &	28.6 &	36.6 &	34.2 \\
 & & BIC &
		1.46 &	1.33 &	30.5 &	67.7 &	\ 1.8 &	\ 0.0 \\ \cline{2-9}
 & & AIC & 
		0.41 &	0.41 &	20.2 &	60.4 &	16.4 &	\ 3.0 \\
 & $(0,0,0.8,0)$ & AIC$_{\rm naive}$ &
		0.53 &	0.50 &	\ 2.4 &	32.8 &	35.1 &	29.7 \\
 & & BIC &
		0.44 &	0.44 &	79.2 &	20.8 &	\ 0.0 &	\ 0.0 \\ \cline{2-9}
 & & AIC & 
		0.38 &	0.29 &	\ 0.2 &	76.3 &	19.7 &	\ 3.8 \\
 & $(0,0,1.4,0)$ & AIC$_{\rm naive}$ &
		0.55 &	0.51 &	\ 0.0 &	32.9 &	35.9 &	31.2 \\
 & & BIC &
		0.35 &	0.25 &	\ 6.2 &	93.4 &	\ 0.4 &	\ 0.0 \\ \hline
\end{tabular}
\end{center}
\label{tab3}
\end{table}

\begin{table}[!p]
\caption{Evaluation of Kullback-Leibler divergence and selection rate of change-points when the regression function is continuous and the number of true change-points is 2: Each value is obtained using the Monte Carlo method for various sample sizes $n$ and true values of the parameter vector $\bm{\theta}^*$.}
\begin{center}
\renewcommand{\baselinestretch}{1.1}\selectfont
\begin{tabular}{cccccccccc} \hline
 & & & \multicolumn{2}{c}{K-L ($\times 100$)} & \multicolumn{5}{c}{selection rate (\%)} \\
 $n$ & $\bm{\theta}^*$ & criterion & mean & med &  0 & 1 & 2 & 3 & 4$^+$ \\ \hline
\multirow{12}{*}{100}	 & & AIC & 
6.50	&	5.10	&	\ 4.0	&	85.0	&	\ 8.1	&	2.4	&	0.5 \\
 & $(0,0,3.5,-1.05,7,-3.5)$ & AIC$_{\rm naive}$ &
8.45	&	7.01	&	31.2	&	68.8	&	\ 0.0	&	0.0	&	0.0 \\
 & & BIC &
7.94	&	6.42	&	26.3	&	73.7	&	\ 0.0	&	0.0	&	0.0 \\
\cline{2-10}
 & & AIC & 
6.84	&	5.58	&	\ 1.3	&	85.5	&	10.3	&	2.1	&	0.8 \\
 & $(0,0,4,-1.2,8,-4)$ & AIC$_{\rm naive}$ &
8.24	&	6.07	&	17.6	&	82.4	&	\ 0.0	&	0.0	&	0.0 \\
 & & BIC &
7.81	&	5.83	&	14.2	&	85.7	&	\ 0.1	&	0.0	&	0.0 \\
\cline{2-10}
 & & AIC & 
7.23	&	5.99	&	\ 0.4	&	82.3	&	14.2	&	2.6	&	0.5 \\
 & $(0,0,4.5,-1.35,9,-4.5)$ & AIC$_{\rm naive}$ &
8.00	&	6.11	&	\ 8.1	&	91.8	&	\ 0.1	&	0.0	&	0.0 \\
 & & BIC &
7.70	&	6.02	&	\ 6.3	&	93.4	&	\ 0.3	&	0.0	&	0.0 \\
\cline{2-10}
 & & AIC & 
8.81	&	7.75	&	33.3	&	19.1	&	41.6	&	5.7	&	0.3 \\
 & $(0,0,5,-1.5,0,2)$ & AIC$_{\rm naive}$ &
8.69	&	7.78	&	89.2	&	\ 6.1	&	\ 4.7	&	0.0	&	0.0 \\
 & & BIC &
8.76	&	7.84	&	85.3	&	\ 7.3	&	\ 7.4	&	0.0	&	0.0 \\
\hline
\multirow{12}{*}{200}	 & & AIC & 
3.58	&	2.96	&	\ 1.1	&	81.9	&	13.9	&	2.4	&	0.7 \\
 & $(0,0,3,-0.9,6,-3	)$ & AIC$_{\rm naive}$ &
3.95	&	3.08	&	10.9	&	88.8	&	\ 0.3	&	0.0	&	0.0 \\
 & & BIC &
4.07	&	3.12	&	12.7	&	87.2	&	\ 0.1	&	0.0	&	0.0 \\
\cline{2-10}
 & & AIC & 
4.14	&	3.27	&	\ 0.0	&	73.9	&	21.8	&	3.1	&	1.2 \\
 & $(0,0,3.5,-1.05,7,-3.5)$ & AIC$_{\rm naive}$ &
3.93	&	3.30	&	\ 2.9	&	96.4	&	\ 0.7	&	0.0	&	0.0 \\
 & & BIC &
4.01	&	3.32	&	\ 3.9	&	95.5	&	\ 0.6	&	0.0	&	0.0 \\
\cline{2-10}
 & & AIC & 
4.10	&	3.54	&	\ 0.0	&	66.3	&	29.5	&	2.9	&	1.3 \\
 & $(0,0,4,-1.2,8,-4	)$ & AIC$_{\rm naive}$ &
4.28	&	3.78	&	\ 0.9	&	97.0	&	\ 2.1	&	0.0	&	0.0 \\
 & & BIC &
4.30	&	3.78	&	\ 1.0	&	97.9	&	\ 1.1	&	0.0	&	0.0 \\
\cline{2-10}
 & & AIC & 
4.53	&	4.05	&	\ 9.3	&	\ 9.5	&	72.8	&	6.8	&	1.6 \\
 & $(0,0,5,-1.5,0,2)$ & AIC$_{\rm naive}$ &
6.67	&	6.71	&	66.7	&	\ 9.9	&	23.2	&	0.2	&	0.0 \\
 & & BIC &
6.83	&	6.74	&	71.6	&	\ 9.9	&	18.4	&	0.1	&	0.0 \\
\hline
\multirow{12}{*}{300}	 & & AIC & 
2.50	&	2.06	&	\ 0.5	&	79.2	&	13.6	&	4.7	&	2.0 \\
 & $(0,0,2.5,-0.75,5,-2.5)$ & AIC$_{\rm naive}$ &
2.63	&	2.12	&	10.4	&	89.3	&	\ 0.3	&	0.0	&	0.0 \\
 & & BIC &
2.82	&	2.20	&	14.3	&	85.4	&	\ 0.3	&	0.0	&	0.0 \\
\cline{2-10}
 & & AIC & 
2.67	&	2.29	&	\ 0.0	&	71.1	&	22.3	&	5.3	&	1.3 \\
 & $(0,0,3,-0.9,6,-3	)$ & AIC$_{\rm naive}$ &
2.66	&	2.33	&	\ 2.1	&	96.9	&	\ 1.0	&	0.0	&	0.0 \\
 & & BIC &
2.73	&	2.34	&	\ 3.2	&	96.4	&	\ 0.4	&	0.0	&	0.0 \\
\cline{2-10}
 & & AIC & 
2.87	&	2.47	&	\ 0.0	&	57.3	&	34.5	&	6.4	&	1.8 \\
 & $(0,0,3.5,-1.05,7,-3.5)$ & AIC$_{\rm naive}$ &
2.96	&	2.71	&	\ 0.0	&	97.7	&	\ 2.3	&	0.0	&	0.0 \\
 & & BIC &
3.01	&	2.72	&	\ 0.4	&	98.3	&	\ 1.3	&	0.0	&	0.0 \\ 
\cline{2-10}
 & & AIC & 
3.02	&	2.52	&	\ 1.8	&	\ 3.3	&	83.4	&	8.5	&	3.0 \\
 & $(0,0,5,-1.5,0,2)$ & AIC$_{\rm naive}$ &
4.75	&	5.55	&	39.6	&	\ 9.3	&	50.8	&	0.3	&	0.0 \\
 & & BIC &
5.30	&	6.24	&	51.0	&	\ 9.5	&	39.4	&	0.1	&	0.0 \\
\hline
\end{tabular}
\end{center}
\label{tab4}
\end{table}

\begin{table}[!p]
\caption{Evaluation of Kullback-Leibler divergence and selection rate of change-points when the regression function is discontinuous and the number of true change-points is 2: Each value is obtained using the Monte Carlo method for various sample sizes $n$ and true values of the parameter vector $\bm{\theta}^*$.}
\begin{center}
\renewcommand{\baselinestretch}{1.1}\selectfont
\begin{tabular}{cccccccccc} \hline
 & & & \multicolumn{2}{c}{K-L ($\times 10$)} & \multicolumn{5}{c}{selection rate (\%)} \\
 $n$ & $\bm{\theta}^*$ & criterion & mean & med &  0 & 1 & 2 & 3 & 4$^+$ \\ \hline
\multirow{12}{*}{100}	 & & AIC & 
1.90	&	1.69	&	38.9	&	22.2	&	30.3	&	\ 7.6	&	\ 1.0
\\
 & $(0,0,1.4,0,2.8,0)$ & AIC$_{\rm naive}$ &
2.23	&	2.98	&	\ 4.8	&	10.1	&	31.7	&	33.5	&	19.9
\\
 & & BIC &
1.70	&	1.58	&	84.3	&	11.5	&	\ 4.1	&	\ 0.1	&	\ 0.0
\\ \cline{2-10}
 & & AIC & 
2.58	&	2.62	&	\ 8.9	&	13.7	&	58.7	&	16.1	&	\ 2.6
\\
 & $(0,0,2,0,4,0)$ & AIC$_{\rm naive}$ &
2.76	&	2.66	&	\ 0.3	&	\ 1.8	&	29.2	&	42.8	&	25.9
\\
 & & BIC &
2.88	&	2.94	&	52.7	&	18.3	&	28.1	&	\ 0.9	&	\ 0.0
\\ \cline{2-10}
 & & AIC & 
1.17	&	1.58	&	\ 0.2	&	67.7	&	24.5	&	\ 6.6	&	\ 1.0
\\
 & $(0,0,1.2,0,0,0)$ & AIC$_{\rm naive}$ &
2.05	&	1.90	&	\ 0.0	&	22.1	&	31.3	&	28.1	&	18.5
\\
 & & BIC &
1.67	&	1.50	&	\ 3.1	&	93.5	&	\ 3.3	&	\ 0.1	&	\ 0.0
\\ \cline{2-10}
 & & AIC & 
2.41	&	2.32	&	\ 0.0	&	32.1	&	51.4	&	14.7	&	\ 1.8
\\
 & $(0,0,1.8,0,0,0)$ & AIC$_{\rm naive}$ &
2.59	&	2.41	&	\ 0.0	&	\ 4.3	&	32.1	&	37.4	&	26.2
\\
 & & BIC &
2.62	&	2.47	&	\ 0.0	&	82.2	&	17.1	&	\ 0.7	&	\ 0.0
\\ \hline
\multirow{12}{*}{200}	 & & AIC & 
1.05	&	1.05	&	18.7	&	18.6	&	49.6	&	11.5	&	\ 1.6
\\
 & $(0,0,1.2,0,2.4,0)$ & AIC$_{\rm naive}$ &
1.12	&	1.09	&	\ 1.6	&	\ 4.2	&	35.8	&	37.3	&	21.1
\\
 & & BIC &
1.14	&	1.10	&	83.7	&	\ 9.3	&	\ 7.0	&	\ 0.0	&	\ 0.0
\\ \cline{2-10}
 & & AIC & 
1.13	&	1.00	&	\ 0.5	&	\ 1.8	&	76.6	&	18.5	&	\ 2.6
\\
 & $(0,0,1.8,0,3.6,0)$ & AIC$_{\rm naive}$ &
1.29	&	1.23	&	\ 0.0	&	\ 0.2	&	34.4	&	40.9	&	24.5
\\
 & & BIC &
1.48	&	1.34	&	24.6	&	12.9	&	62.1	&	\ 0.4	&	\ 0.0
\\ \cline{2-10}
 & & AIC & 
0.94	&	0.87	&	\ 0.0	&	56.1	&	35.9	&	\ 6.5	&	\ 1.5
\\
 & $(0,0,1,0,0,0)$ & AIC$_{\rm naive}$ &
1.07	&	1.01	&	\ 0.0	&	14.6	&	33.8	&	32.6	&	19.0
\\
 & & BIC &
0.92	&	0.85	&	\ 0.7	&	97.8	&	\ 1.5	&	\ 0.0	&	\ 0.0
\\ \cline{2-10}
 & & AIC & 
1.10	&	1.03	&	\ 0.0	&	\ 8.9	&	71.6	&	17.7	&	\ 1.8
\\
 & $(0,0,1.6,0,0,0)$ & AIC$_{\rm naive}$ &
1.23	&	1.17	&	\ 0.0	&	\ 0.6	&	34.6	&	40.8	&	24.0
\\
 & & BIC &
1.39	&	1.47	&	\ 0.0	&	57.6	&	41.9	&	\ 0.5	&	\ 0.0
\\ \hline
\multirow{12}{*}{300}	 & & AIC & 
0.69	&	0.71	&	17.3	&	17.9	&	51.7	&	11.7	&	\ 1.4
\\
 & $(0,0,1,0,2,0)$ & AIC$_{\rm naive}$ &
0.74	&	0.72	&	\ 1.3	&	\ 3.6	&	33.8	&	40.3	&	21.0
\\
 & & BIC &
0.78	&	0.76	&	85.4	&	\ 8.6	&	\ 5.9	&	\ 0.1	&	\ 0.0
\\ \cline{2-10}
 & & AIC & 
0.69	&	0.61	&	\ 0.0	&	\ 0.5	&	77.5	&	19.4	&	\ 2.6
\\
 & $(0,0,1.6,0,3.2,0)$ & AIC$_{\rm naive}$ &
0.82	&	0.77	&	\ 0.0	&	\ 0.0	&	32.2	&	41.8	&	26.0
\\
 & & BIC &
0.90	&	0.66	&	16.3	&	\ 9.3	&	74.0	&	\ 0.4	&	\ 0.0
\\ \cline{2-10}
 & & AIC & 
0.61	&	0.57	&	\ 0.0	&	60.4	&	30.5	&	\ 7.9	&	\ 1.2
\\
 & $(0,0,0.8,0,0,0)$ & AIC$_{\rm naive}$ &
0.71	&	0.68	&	\ 0.0	&	13.8	&	35.5	&	32.3	&	18.4
\\
 & & BIC &
0.60	&	0.56	&	\ 1.2	&	98.1	&	\ 0.7	&	\ 0.0	&	\ 0.0
\\ \cline{2-10}
 & & AIC & 
0.66	&	0.60	&	\ 0.0	&	\ 3.2	&	77.2	&	18.0	&	\ 1.6
\\
 & $(0,0,1.4,0,0,0)$ & AIC$_{\rm naive}$ &
0.77	&	0.73	&	\ 0.0	&	\ 0.0	&	34.5	&	41.5	&	24.0
\\
 & & BIC &
0.91	&	1.01	&	\ 0.0	&	51.4	&	48.6	&	\ 0.0	&	\ 0.0 
\\ \hline
\end{tabular}
\end{center}
\label{tab5}
\end{table}

\section{Real data analysis}\label{sec6}

In Sections \ref{sec4} and \ref{sec5}, AIC was derived separately for models with change-points that leave the regression function continuous and for models with change-points that make it discontinuous; it can be extended to models with both types of change-points.
That is, when the number of change-points in the continuous case is $m_1$ and the number of change-points in the discontinuous case is $m_2$ ($m=m_1+m_2$), we only have to set the penalty term for AIC to $2p+2m_1p+2m_2(3+p)=2p(m+1)+6m_2$.
By using this for model selection, it is possible to estimate not only the number of change-points, but also whether the regression function is continuous or discontinuous at each change-point.
The information criteria to be compared are AIC$_{\rm navie}$ and BIC.
Since AIC$_{\rm navie}$ in Tables \ref{tab2} and \ref{tab4}, which sets the penalty of the change-point parameter to 6 even in the continuous case, is close to BIC, we leave it at 2 and treat AIC$_{\rm navie}$, which sets the penalty of the change-point parameter to 2 even in the discontinuous case.
That is, the penalty term for AIC$_{\rm navie}$ is $2p+2m_1p+2m_2(1+p)=2p(m+1)+2m_2$ and that for BIC is $\{p+m_1p+m_2(1+p)\}\log n=\{p(m+1)+m_2\}\log n$.

The real data analyzed are the number of new COVID-19 infections in the United Kingdom for the 240-day period from March 13, 2020 to November 7, 2020.
Supposing a piecewise-linear segmented regression model in \eqref{model}, we have estimated the change-points and the regression function drawn in the left panel of Figure \ref{fig1}.
In this setting, both AIC$_{\rm navie}$ and BIC give the same results as AIC.
In the right panel, the regression function is supposed to be discontinuous at the change-points, as in \cite{JiaZS22}.
In fact, the result is close to that given by \cite{JiaZS22}.
In this setting, while BIC gives the same result, AIC$_{\rm navie}$ selects more change-points.
Since these are not artificial data, we cannot say whether the left panel or the right panel is closer to the true structure; however, it is more natural that they should be continuous unless there is some reason, such as a change in the way of data collection, and therefore we can infer that some of them are in fact continuous as in the left panel.

Considering the possibility that large fluctuations in the data may make analysis difficult, we reduce the data size to 200 and conduct the estimation in Figure \ref{fig2}.
In these data, while AIC$_{\rm navie}$ gives the same results, BIC selects fewer change-points.
The true structure is still unknown; it is highly possible that there is a discontinuity in the regression function in the later part of the data for some reason.
In order to find out where and how large the discontinuities actually are, we will need to do a more sophisticated modeling adapted to the data.
What we want to insist on in this section is that the results will differ considerably depending on the model selection criteria and that it is therefore important to derive a reasonable criterion.

\begin{figure}[!t]
\begin{overpic}[scale=0.44]{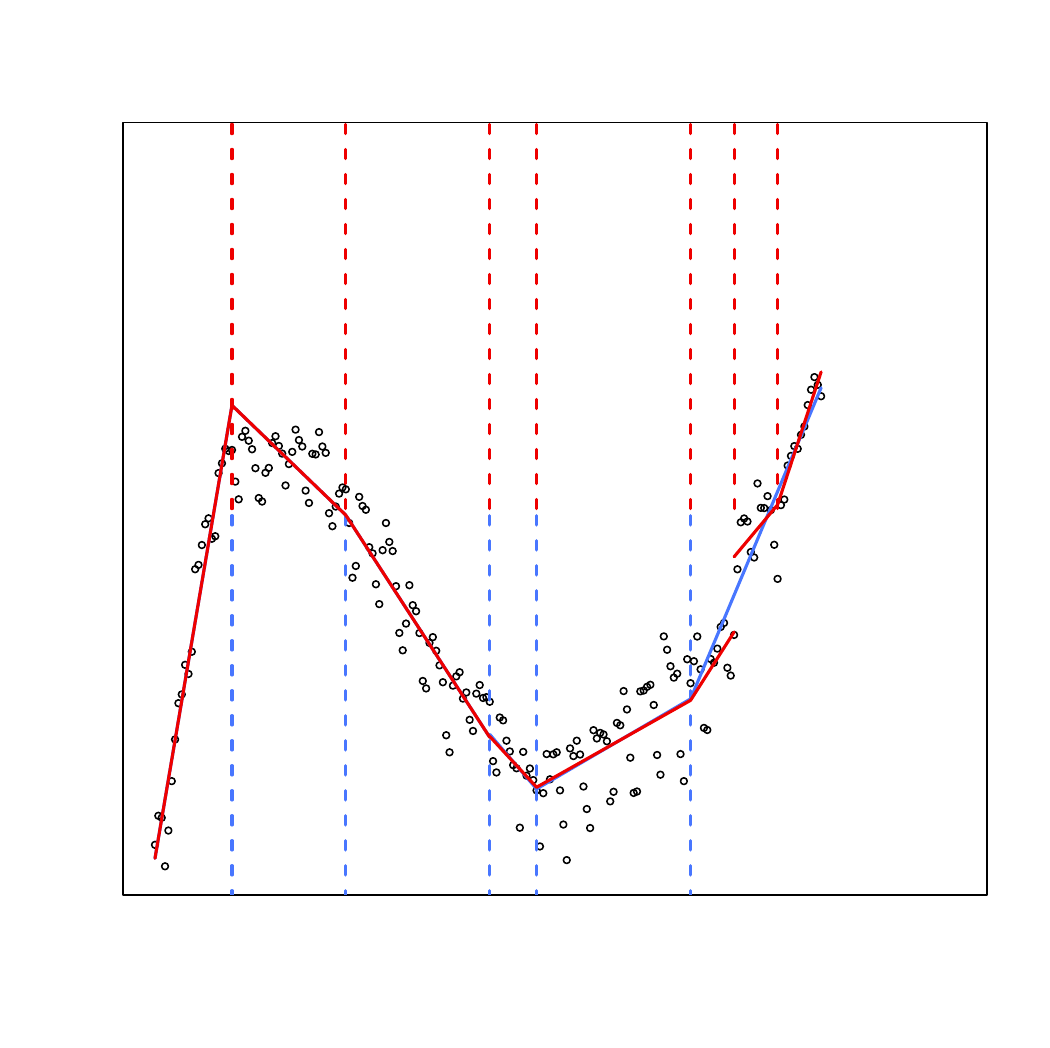}
\put(0,28){\footnotesize\rotatebox{90}{logarithm of new cases}}
\put(6,18){\scriptsize\rotatebox{90}{6}}
\put(6,33){\scriptsize\rotatebox{90}{7}}
\put(6,48){\scriptsize\rotatebox{90}{8}}
\put(6,63){\scriptsize\rotatebox{90}{9}}
\put(6,78){\scriptsize\rotatebox{90}{10}}
\put(9,9){\scriptsize Mar-13}
\put(28,9){\scriptsize May-12}
\put(47,9){\scriptsize Jul-11}
\put(66,9){\scriptsize Sep-09}
\put(85,9){\scriptsize Nov-07}
\put(50,3){\footnotesize date}
\end{overpic}
\hspace{5mm}
\begin{overpic}[scale=0.44]{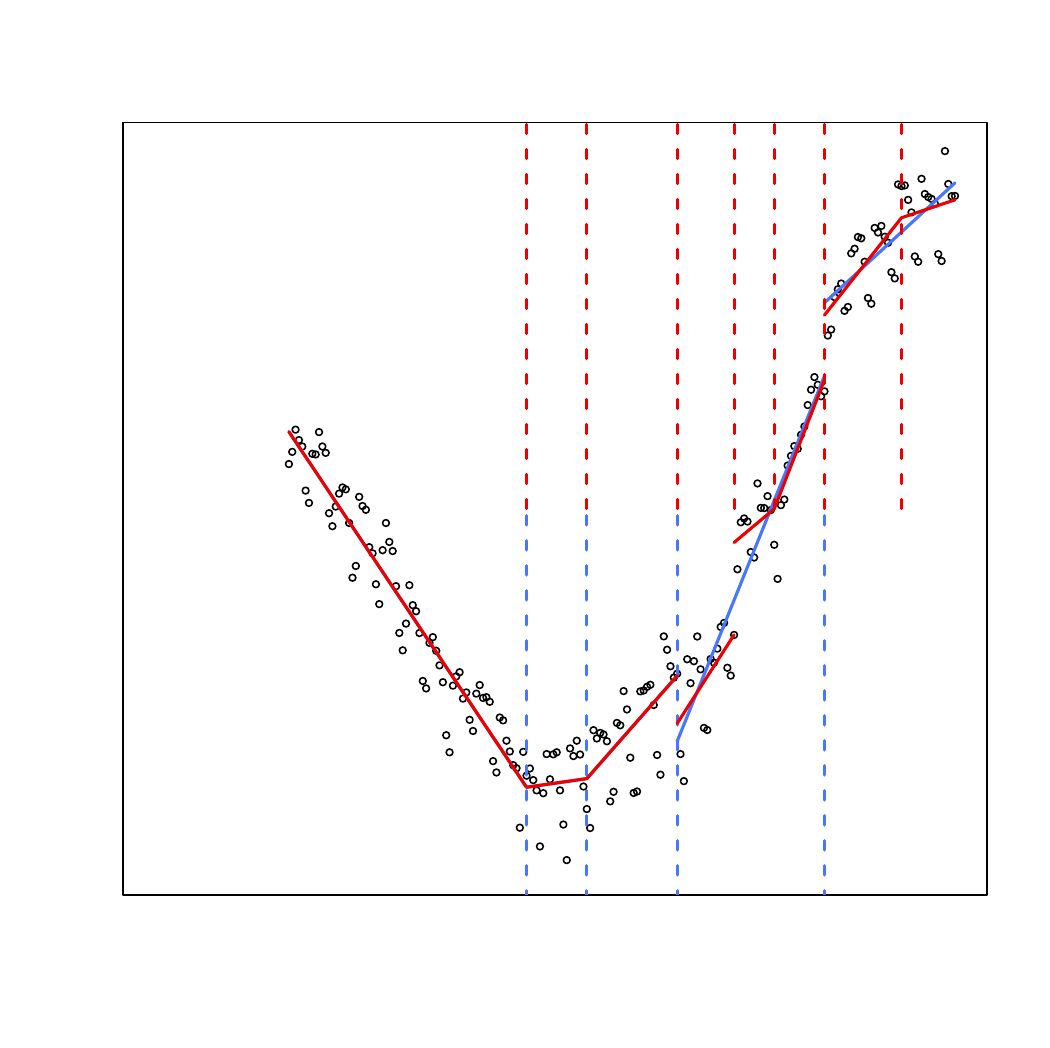}
\put(0,28){\footnotesize\rotatebox{90}{logarithm of new cases}}
\put(6,18){\scriptsize\rotatebox{90}{6}}
\put(6,33){\scriptsize\rotatebox{90}{7}}
\put(6,48){\scriptsize\rotatebox{90}{8}}
\put(6,63){\scriptsize\rotatebox{90}{9}}
\put(6,78){\scriptsize\rotatebox{90}{10}}
\put(9,9){\scriptsize Mar-13}
\put(28,9){\scriptsize May-12}
\put(47,9){\scriptsize Jul-11}
\put(66,9){\scriptsize Sep-09}
\put(85,9){\scriptsize Nov-07}
\put(50,3){\footnotesize date}
\end{overpic}
\caption{Results obtained by allowing the regression function to be continuous at the change-points for partial data on the number of new COVID-19 infections in the United Kingdom: The red solid and broken lines are the regression function and change-points estimated using the proposed AIC. The blue solid and broken lines are the regression function and change-points estimated using BIC.}
\label{fig2}
\end{figure}
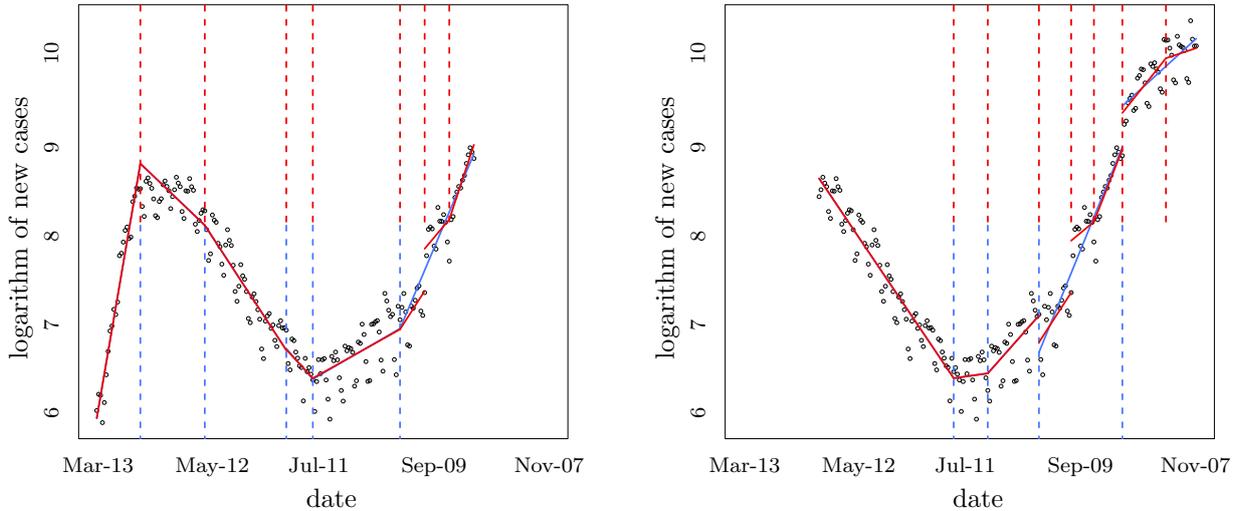

\section{Summary and discussion}\label{sec7}

Information criteria in segmented regression have so far been developed mainly for BIC-type with model selection consistency from the perspective of increasing the probability of selecting the true model.
In this paper, we have derived an information criterion based on the original definition of AIC from the perspective of making the divergence between the true and estimated structures small.
The reason why this result has not been obtained until now is the existence of asymptotics specific to segmented regression, and we have actually confirmed that when the regression function is discontinuous at the change-points, the penalty for the change-point parameter is 6, as in the result of \cite{Nin15}.
When the regression function is continuous at the change-points, the asymptotic normality of the change-point estimators has been derived first by developing \cite{KimK08}.
Its derivation basically follows the flow of \cite{KimK08}; however, because the model is extended to generalized nonlinear ones, it is necessary to change how to use the continuity and to add some assumption.
Then it has been shown that the penalty on the change-point parameter remains 2 despite the abrupt change in the differential coefficients.
In numerical experiments, we have compared the derived AIC with AIC$_{\rm naive}$, which incorrectly captures the peculiarities of the model, and BIC, and confirmed that the derived AIC tends to make the divergence small, as we expected.
In real data analysis, we have treated the data on the number of new COVID-19 infections, which gave the motivation for this study, and judged that some of the jumps in the regression function estimated by \cite{JiaZS22} would be connected according to our AIC.
We have also confirmed that while the results of the three criteria happened to coincide when the data were analyzed as they were, they returned different results when the time period of the data was slightly changed, that is, the difference in the criteria was not negligible.

One of the important settings in which our method can be extended in the near future will be segmented regression in survival analysis.
While \cite{LiuLS08} and \cite{OzaN23} apply a Cox model with jumps in the regression function to data from randomized, placebo-controlled clinical trials of patients with malignant glioma, changes without jumps can naturally be expected in such data.
Although the model is not included in \eqref{model2} and the partial likelihood is not treated in this paper, it is expected that a risk evaluation similar to that of \cite{OzaN23} can be performed.
Another near future issue is to extend our method to more recent change-point problems of importance, such as those listed in Section \ref{sec1} (\citealt{AueH13}, \citealt{BerGHK09}, and \citealt{AueHHR09}).
The same derivation should be possible for serially correlated data, and it is expected that for functional and high-dimensional data, the number of change-points as well as the number of principal components, can be estimated by developing $C_p$ criterion.
A further challenge is the extension of the method for selecting the number of nodes in a spline by \cite{Rup02}, with a view to contributing to the ongoing topic of selecting the number of hyper-parameters in nonparametric regression (e.g., \citealt{WooPS16}). 
Considering that jumps are allowed in the framework of generalized linear models, it is necessary to develop AIC as in this paper.
Since derivative coefficients vary smoothly in jump-free splines, our method must be extended from this stage, and since regularization estimation is essential, it will be necessary to incorporate methods such as the one in \cite{NinK16}.

\section*{Appendix}

\subsection*{Proof of Lemma \ref{lem1}}

Since $a(\cdot)$ and $\phi(\cdot;\cdot)$ are $C^2$-functions, we have
\begin{align}
& y_i\phi(x_i;\bm{\theta})-a(\phi(x_i;\bm{\theta}))
\notag \\
& = y_i\phi(x_i;\bm{\theta}^*)-a(\phi(x_i;\bm{\theta}^*))+(\phi(x_i;\bm{\theta})-\phi(x_i;\bm{\theta}^*))(y_i-a'(\phi(x_i;\bm{\theta}^*)))
\notag \\
& \ \phantom{=} -\frac{1}{2}(\phi(x_i;\bm{\theta})-\phi(x_i;\bm{\theta}^*))^2a''(\phi(x_i;\bm{\theta}^*))+\oP(\|\bm{\theta}-\bm{\theta}^*\|^2)
\label{forlem2}
\end{align}
for $\bm{\theta},\ \bm{\theta}^*\in\mathbb{R}^p$. 
Therefore, if $\bm{\theta}\to\bm{\theta}^*$, it holds
\begin{align*}
H(\bm{\theta}) & = \frac{1}{n}\sum_{k=1}^{m+1}\sum_{i=1}^nI_{[\tau^{[k-1]*},\tau^{[k]*})}(x_i)
\\
& \ \phantom{=} \times \bigg\{(\phi(x_i;\bm{\theta}^{[k]})-\phi(x_i;\bm{\theta}^{[k]*}))(y-a'(\phi(x_i;\bm{\theta}^{[k]*})))
\\
& \ \phantom{= \bigg[ \times} - \frac{1}{2}(\phi(x_i;\bm{\theta}^{[k]})-\phi(x_i;\bm{\theta}^{[k]*}))^2a''(\phi(x_i;\bm{\theta}^{[k]*})) + \oP(\|\bm{\theta}^{[k]}-\bm{\theta}^{[k]*}\|^2)\bigg\}.
\end{align*}
Now, using
\begin{align*}
& \phi(x_i;\bm{\theta}^{[k]})-\phi(x_i;\bm{\theta}^{[k]*})
\\
& = (\bm{\theta}^{[k]}-\bm{\theta}^{[k]*})^\T\frac{\partial}{\partial\bm{\theta}}\phi(x_i;\bm{\theta}^{[k]*}) + (\bm{\theta}^{[k]}-\bm{\theta}^{[k]*})^\T\frac{\partial^2}{\partial\bm{\theta}\partial\bm{\theta}^\T}\phi(x_i;\bm{\theta}^{[k]*})(\bm{\theta}^{[k]}-\bm{\theta}^{[k]*})
\\
& \ \phantom{=} + \oP(\|\bm{\theta}^{[k]}-\bm{\theta}^{[k]*}\|^2)
\end{align*}
and denoting
\begin{align*}
& \bm{\Delta}_i^{[k]} = I_{[\tau^{[k-1]*},\tau^{[k]*})}(x_i)(y_i-a'(\phi(x_i;\bm{\theta}^{[k]*})))\frac{\partial}{\partial\bm{\theta}}\phi(x_i;\bm{\theta}^{[k]*}),
\\
& \bm{\Delta}_i = (\bm{\Delta}_i^{[1]\T},\bm{\Delta}_i^{[2]\T},\ldots,\bm{\Delta}_i^{[m+1]\T})^{\T},
\\
& \bm{S}_i^{[k]} = I_{[\tau^{[k-1]*},\tau^{[k]*})}(x_i)\bigg[a''(\phi(x_i;\bm{\theta}^{[k]*}))\frac{\partial}{\partial\bm{\theta}}\phi(x_i;\bm{\theta}^{[k]*})\frac{\partial}{\partial\bm{\theta}^{\T}}\phi(x_i;\bm{\theta}^{[k]*})
\\
& \phantom{\bm{S}_i^{[k]} = I_{[\tau^{[k-1]*},\tau^{[k]*})}(x_i)\bigg[} + (y_i-a'(\phi(x_i;\bm{\theta}^{[k]*})))\frac{\partial^2}{\partial\bm{\theta}\partial\bm{\theta}^\T}\phi(x_i;\bm{\theta}^{[k]*})\bigg],
\\
& \bm{S}_i = {\rm diag}(\bm{S}_i^{[1]},\bm{S}_i^{[2]},\ldots,\bm{S}_i^{[m+1]}),
\end{align*}
it can be rewritten that
\begin{align*}
H(\bm{\theta}) = (\bm{\theta}-\bm{\theta}^*)^\T\frac{1}{n}\sum_{i=1}^n\bm{\Delta}_i-\frac{1}{2}(\bm{\theta}-\bm{\theta}^*)^\T \bigg(\frac{1}{n}\sum_{i=1}^n\bm{S}_i\bigg) (\bm{\theta}-\bm{\theta}^*) + \oP(\|\bm{\theta}-\bm{\theta}^*\|^2).
\end{align*}
Furthermore, since it holds
\begin{align*}
& \E[\bm{\Delta}_i^{[k]}] = \E[\E[\bm{\Delta}_i^{[k]}\mid x_i]]
\\
& = \E\bigg[I_{[\tau^{[k-1]*},\tau^{[k]*})}(x_i)\E[y_i-a'(\phi(x_i;\bm{\theta}^{[k]*}))\mid x_i]\frac{\partial}{\partial\bm{\theta}}\phi(x_i;\bm{\theta}^{[k]*})\bigg] = \bm{0}_p, 
\\
& \E[\bm{\Delta}_i^{[k]}\bm{\Delta}_i^{[k]\T}] = \E[\E[\bm{\Delta}_i^{[k]}\bm{\Delta}_i^{[k]\T}\mid x_i]]
\\
& = \E\bigg[I_{[\tau^{[k-1]*},\tau^{[k]*})}(x_i)\E[\{y_i-a'(\phi(x_i;\bm{\theta}^{[k]*}))\}^2\mid x_i]\frac{\partial}{\partial\bm{\theta}}\phi(x_i;\bm{\theta}^{[k]*})\frac{\partial}{\partial\bm{\theta}^{\T}}\phi(x_i;\bm{\theta}^{[k]*})\bigg] = \bm{T}^{[k]},
\\
& \E[\bm{S}_i^{[k]}] = \E[\E[\bm{S}_i^{[k]}\mid x_i]] = \bm{T}^{[k]},
\end{align*}
the central limit theorem and the law of large numbers complete the proof.

\subsection*{Proof of \eqref{gtoinf}}

As a representative of the situation $\bm{\tau}-\bm{\tau}^*\neq\O(n^{-1})$, we consider the case that $\tau^{[k^{\dagger}]}-\tau^{[k^{\dagger }]*}\neq\O(n^{-1})$ and $\tau^{[k^{\dagger}]*}<\tau^{[k^{\dagger}]}$ for some $k^{\dagger}\in\{1,2,\ldots,m\}$ and $\tau^{[k]}-\tau^{[k]*}=\O(n^{-1})$ for all $k\in\{1,2,\ldots,m\}\setminus k^{\dagger}$. 
Note that the following story holds for other cases where $\bm{\tau}-\bm{\tau}^*\neq\O(n^{-1})$.
In this representative case, it holds
\begin{align*}
& \sum_{i=1}^ng(\hat{\bm{\theta}}_{\bm{\tau}},\bm{\tau};x_i,y_i)-\sum_{i=1}^ng(\hat{\bm{\theta}}_{\bm{\tau}^*},\bm{\tau}^*;x_i,y_i)
\\
& = \sum_{i:\tau^{[k^{\dagger}-1]}<x_i\le\tau^{[k^{\dagger}]*}}\{y_i(\phi(x_i;\hat{\bm{\theta}}_{\bm{\tau}}^{[k^{\dagger}]})-\phi(x_i;\hat{\bm{\theta}}_{\bm{\tau}^*}^{[k^{\dagger}]}))\}-a(\phi(x_i;\hat{\bm{\theta}}_{\bm{\tau}}^{[k^{\dagger}]}))+a(\phi(x_i;\hat{\bm{\theta}}_{\bm{\tau}^*}^{[k^{\dagger}]}))\}
\\
& \ \phantom{=} + \sum_{i:\tau^{[k^{\dagger}]*}<x_i\le\tau^{[k^{\dagger}]}}\{y_i(\phi(x_i;\hat{\bm{\theta}}_{\bm{\tau}}^{[k^{\dagger}]}-\phi(x_i;\hat{\bm{\theta}}_{\bm{\tau}^*}^{[k^{\dagger}+1]}))-a(\phi(x_i;\hat{\bm{\theta}}_{\bm{\tau}}^{[k^{\dagger}]}))+a(\phi(x_i;\hat{\bm{\theta}}_{\bm{\tau}^*}^{[k^{\dagger}+1]}))\}
\\
& \ \phantom{=} + \sum_{i:\tau^{[k^{\dagger}]}<x_i\le\tau^{[k^{\dagger}+1]}}\{y_i(\phi(x_i;\hat{\bm{\theta}}_{\bm{\tau}}^{[k^{\dagger}+1]})-\phi(x_i;\hat{\bm{\theta}}_{\bm{\tau}^*}^{[k^{\dagger}+1]}))\}-a(\phi(x_i;\hat{\bm{\theta}}_{\bm{\tau}}^{[k^{\dagger}+1]}))+a(\phi(x_i;\hat{\bm{\theta}}_{\bm{\tau}^*}^{[k^{\dagger}+1]}))
\\
& \ \phantom{=} + \OP(1)
\\
& = \sum_{i:\tau^{[k^{\dagger}-1]}<x_i\le\tau^{[k^{\dagger}]*}}[\{\phi(x_i;\hat{\bm{\theta}}_{\bm{\tau}}^{[k^{\dagger}]})-\phi(x_i;\hat{\bm{\theta}}_{\bm{\tau}^*}^{[k^{\dagger}]})\}\{y_i-a'(\phi(x_i;\hat{\bm{\theta}}_{\bm{\tau}^*}^{[k^{\dagger}]}))\}-B_{\hat{\bm{\theta}}_{\bm{\tau}}^{[k^{\dagger}]},\hat{\bm{\theta}}_{\bm{\tau}^*}^{[k^{\dagger}]}}]
\\
& \ \phantom{=} + \sum_{i:\tau^{[k^{\dagger}]*}<x_i\le\tau^{[k^{\dagger}]}}[(\phi(x_i;\hat{\bm{\theta}}_{\bm{\tau}}^{[k^{\dagger}]})-\phi(x_i;\hat{\bm{\theta}}_{\bm{\tau}^*}^{[k^{\dagger}+1]}))\{y_i-a'(\phi(x_i;\hat{\bm{\theta}}_{\bm{\tau}^*}^{[k^{\dagger}+1]}))\}-B_{\hat{\bm{\theta}}_{\bm{\tau}}^{[k^{\dagger}]},\hat{\bm{\theta}}_{\bm{\tau}^*}^{[k^{\dagger}+1]}}]
\\
& \ \phantom{=} + \sum_{i:\tau^{[k^{\dagger}]}<x_i\le\tau^{[k^{\dagger}+1]}}[\{\phi(x_i;\hat{\bm{\theta}}_{\bm{\tau}}^{[k^{\dagger}+1]})-\phi(x_i;\hat{\bm{\theta}}_{\bm{\tau}^*}^{[k^{\dagger}+1]})\}\{y_i-a'(\phi(x_i;\hat{\bm{\theta}}_{\bm{\tau}^*}^{[k^{\dagger}+1]}))\}-B_{\hat{\bm{\theta}}_{\bm{\tau}}^{[k^{\dagger}+1]},\hat{\bm{\theta}}_{\bm{\tau}^*}^{[k^{\dagger}+1]}}]
\\
& \ \phantom{=} + \OP(1), 
\end{align*}
where $B_{\bm{\theta}^\dag,\bm{\theta}^\ddag}=a(\phi(x_i;\bm{\theta}^{\dag}))-a(\phi(x_i;\bm{\theta}^{\ddag}))-(\phi(x_i;\bm{\theta}^\dag)-\phi(x_i;\bm{\theta}^\ddag))a'(\phi(x_i;\bm{\theta}^{\ddag}))$.
Note that the first $\OP(1)$ contains a term from \eqref{order1}.
Using the same derivation as for \eqref{keyasym}, while taking care of Assumption \eqref{asm1} for $x_i$, we see that $\hat{\bm{\theta}}_{\bm{\tau}}^{[k^{\dagger}]}-\hat{\bm{\theta}}_{\bm{\tau}^*}^{[k^{\dagger}]}$ and $\hat{\bm{\theta}}_{\bm{\tau}}^{[k^{\dagger}+1]}-\hat{\bm{\theta}}_{\bm{\tau}^*}^{[k^{\dagger}+1]}$ are $\OP(\tau^{[k^{\dagger}]}-\tau^{[k^{\dagger}]*})$. 
Then, it follows from the central limit theorem that $\sum_{i:\tau^{[k^{\dagger}-1]}<x_i\le\tau^{[k^{\dagger}]*}}(\phi(x_i;\allowbreak\hat{\bm{\theta}}_{\bm{\tau}}^{[k^{\dagger}]})-\phi(x_i;\hat{\bm{\theta}}_{\bm{\tau}^*}^{[k^{\dagger}]}))(y_i-a'(\phi(x_i;\hat{\bm{\theta}}_{\bm{\tau}^*}^{[k^{\dagger}]})))+\sum_{i:\tau^{[k^{\dagger}]}<x_i\le\tau^{[k^{\dagger}+1]}}(\phi(x_i;\hat{\bm{\theta}}_{\bm{\tau}}^{[k^{\dagger}+1]})-\phi(x_i;\hat{\bm{\theta}}_{\bm{\tau}^*}^{[k^{\dagger}+1]}))(y_i-a'(\phi(x_i;\allowbreak\hat{\bm{\theta}}_{\bm{\tau}^*}^{[k^{\dagger}+1]})))$ is $\OP(\sqrt{n}(\tau^{[k^{\dagger}]}-\tau^{[k^{\dagger}]*})^{1/2})$, and from the definition that $\sum_{i:\tau^{[k^{\dagger}-1]}<x_i\le\tau^{[k^{\dagger}]*}}\allowbreak B_{\hat{\bm{\theta}}_{\bm{\tau}}^{[k^{\dagger}]},\hat{\bm{\theta}}_{\bm{\tau}^*}^{[k^{\dagger}]}}+\allowbreak\sum_{i:\tau^{[k^{\dagger}]*}<x_i\le\tau^{[k^{\dagger}]}}B_{\hat{\bm{\theta}}_{\bm{\tau}}^{[k^{\dagger}]},\hat{\bm{\theta}}_{\bm{\tau}^*}^{[k^{\dagger}+1]}}+\sum_{i:\tau^{[k^{\dagger}]}<x_i\le\tau^{[k^{\dagger}+1]}}B_{\hat{\bm{\theta}}_{\bm{\tau}}^{[k^{\dagger}+1]},\hat{\bm{\theta}}_{\bm{\tau}^*}^{[k^{\dagger}+1]}}$ is positive and not $\oP(n(\tau^{[k^{\dagger}]}-\tau^{[k^{\dagger}]*}))$ since $\phi(x_i;\hat{\bm{\theta}}_{\bm{\tau}}^{[k^{\dagger}]})-\phi(x_i;\hat{\bm{\theta}}_{\bm{\tau}^*}^{[k^{\dagger}+1]})\neq\oP(1)$.
Thus, we obtain
\begin{align}
\P\bigg(\sum_{i=1}^ng(\hat{\bm{\theta}}_{\bm{\tau}},\bm{\tau};x_i,y_i)-\sum_{i=1}^ng(\hat{\bm{\theta}}_{\bm{\tau}^*},\bm{\tau}^*;x_i,y_i)>-M\bigg)\to 0. 
\label{infzero}
\end{align}
The proof is complete from \eqref{order2} and \eqref{infzero}.

\subsection*{Proof of Theorem \ref{thm4}}

When $\bm{\tau}-\bm{\tau}^*=\O(n^{-1})$, it follows from \eqref{keyasym} and Taylor expansion that
\begin{align}
& \sum_{i=1}^ng(\bm{\theta}^*,\bm{\tau};\tilde{x}_i,\tilde{y}_i) - \sum_{i=1}^ng(\hat{\bm{\theta}}_{\bm{\tau}},\bm{\tau};\tilde{x}_i,\tilde{y}_i)
\notag \\
& = -\sum_{k=1}^{m+1}\sum_{i:\tau^{[k-1]}<\tilde{x}_i\le\tau^{[k]}}(\hat{\bm{\theta}}_{\bm{\tau}}^{[k]}-\bm{\theta}^{[k]*})^{\T}g'(\bm{\theta}^{[k]*};\tilde{x}_i,\tilde{y}_i)
\notag \\
& \ \phantom{=} + \frac{1}{2}\sum_{k=1}^{m+1}\sum_{i:\tau^{[k-1]}<\tilde{x}_i\le\tilde{\tau}^{[k]}}(\hat{\bm{\theta}}_{\bm{\tau}}^{[k]}-\bm{\theta}^{[k]*})^{\T}g''(\bm{\theta}^{[k]*};\tilde{x}_i,\tilde{y}_i)(\hat{\bm{\theta}}_{\bm{\tau}}^{[k]}-\bm{\theta}^{[k]*}) + \oP(1)
\notag \\
& = -\sum_{k=1}^{m+1}\sum_{i:\tau^{[k-1]*}<\tilde{x}_i\le\tau^{[k]*}}(\hat{\bm{\theta}}_{\bm{\tau}}^{[k]}-\bm{\theta}^{[k]*})^{\T}g'(\bm{\theta}^{[k]*};\tilde{x}_i,\tilde{y}_i)
\notag \\
& \ \phantom{=} + \frac{1}{2}\sum_{k=1}^{m+1}\sum_{i:\tau^{[k-1]*}<\tilde{x}_i\le\tau^{[k]*}}(\hat{\bm{\theta}}_{\bm{\tau}}^{[k]}-\bm{\theta}^{[k]*})^{\T}g''(\bm{\theta}^{[k]*};\tilde{x}_i,\tilde{y}_i)(\hat{\bm{\theta}}_{\bm{\tau}}^{[k]}-\bm{\theta}^{[k]*}) + \oP(1)
\notag \\
& = -\sum_{k=1}^{m+1}\sum_{i:\tau^{[k-1]*}<\tilde{x}_i\le\tau^{[k]*}}(\hat{\bm{\theta}}_{\bm{\tau}^*}^{[k]}-\bm{\theta}^{[k]*})^{\T}g'(\bm{\theta}^{[k]*};\tilde{x}_i,\tilde{y}_i)
\notag \\
& \ \phantom{=} + \frac{1}{2}\sum_{k=1}^{m+1}\sum_{i:\tau^{[k-1]*}<\tilde{x}_i\le\tau^{[k]*}}(\hat{\bm{\theta}}_{\bm{\tau}^*}^{[k]}-\bm{\theta}^{[k]*})^{\T}g''(\bm{\theta}^{[k]*};\tilde{x}_i,\tilde{y}_i)(\hat{\bm{\theta}}_{\bm{\tau}^*}^{[k]}-\bm{\theta}^{[k]*}) + \oP(1)
\notag \\
& = -\sum_{k=1}^{m+1}\bm{u}^{[k]\T}\tilde{\bm{u}}^{[k]}+\frac{1}{2}\sum_{k=1}^{m+1}\bm{u}^{[k]\T}\bm{u}^{[k]} + \oP(1),
\label{thm4for1}
\end{align}
where $\tilde{\bm{u}}^{[k]}$ is a copy of $\bm{u}^{[k]}$, that is, a random vector following a multi-dimensional Gaussian distribution $\N(\bm{0}_p,\bm{I}_{p\times p})$ independently of $\bm{u}^{[k]}$.
In addition, it follows from \eqref{supres2} that
\begin{align}
\sum_{i=1}^ng\bigg(\bm{\theta}^*,\argsup_{\bm{\tau}}\sum_{i^{\dagger}=1}^n\hat{g}(\bm{\tau};x_{i^{\dagger}},y_{i^{\dagger}});\tilde{x}_i,\tilde{y}_i\bigg)
= \sum_{k=1}^{m}\tilde{Q}^{[k]}(\check{\tau}^{[k]}) + \oP(1),
\label{thm4for2}
\end{align}
where $\tilde{Q}^{[k]}(\cdot)$ is a copy of $Q^{[k]}(\cdot)$.
Thus, we obtain from \eqref{thm4for1} and \eqref{thm4for2} that
\begin{align}
& \sum_{i=1}^n\hat{g}\bigg(\argsup_{\bm{\tau}}\sum_{i^{\dagger}=1}^n\hat{g}(\bm{\tau};x_{i^{\dagger}},y_{i^{\dagger}});\tilde{x}_i,\tilde{y}_i\bigg)
\notag \\
& = \sum_{i=1}^ng\bigg(\bm{\theta}^*,\argsup_{\bm{\tau}}\sum_{i=1}^n\hat{g}(\bm{\tau};x_i,y_i);\tilde{x}_i,\tilde{y}_i\bigg) + \sum_{k=1}^{m+1}\bm{u}^{[k]\T}\tilde{\bm{u}}^{[k]} - \frac{1}{2}\sum_{k=1}^{m+1}\bm{u}^{[k]\T}\bm{u}^{[k]} + \oP(1)
\notag \\
& = \sum_{k=1}^m\tilde{Q}^{[k]}(\check{\tau}^{[k]}) + \sum_{k=1}^{m+1}\bm{u}^{[k]\T}\tilde{\bm{u}}^{[k]} - \frac{1}{2}\sum_{k=1}^{m+1}\bm{u}^{[k]\T}\bm{u}^{[k]} + \oP(1).
\label{supres3}
\end{align}
The proof is complete from \eqref{supres1} and \eqref{supres3}.

\subsection*{Proof of Theorem \ref{thm5}}

Regarding the asymptotic behavior of the change-point estimator under Assumption \ref{asm2}, the same argument as before applies: when $\bm{\tau}-\bm{\tau}^*=\O(n^{-1}\alpha_n)$, we can show $\sum_{i=1}^ng(\hat{\bm{\theta}}_{\bm{\tau}},\bm{\tau};x_i,y_i)=\OP(1)$, and when $\bm{\tau}-\bm{\tau}^* \neq \O(n^{-1}\alpha_n)$, we can show \eqref{gtoinf}.
Therefore, we take $\sup_{\bm{\tau}}$ and $\argsup_{\bm{\tau}}$ in the set of $\bm{\tau}$ such that $\bm{\tau}-\bm{\tau}^*=\O(n^{-1}\alpha_n)$.
Then, by Taylor expansion, $Q^{[k]}(\tau^{[k]*}+sn^{-1}\alpha_n)$ can be written as
\begin{align*}
& I_{\{s<0\}}\bigg(\frac{1}{\sqrt{\alpha_n}}\sum_{i:\tau^{[k]*}+s\alpha_n/n<x_i\le\tau^{[k]*}}\bm{\Delta}_{\bm{\theta}^*}^{[k]\T}g'(\bm{\theta}^{[k]*};x_i,y_i)
\\
& \ \phantom{I_{\{s<0\}}\bigg(} + \frac{1}{2\alpha_n}\sum_{i:\tau^{[k]*}+s\alpha_n/n<x_i\le\tau^{[k]*}}\bm{\Delta}_{\bm{\theta}^*}^{[k]\T}g''(\bm{\theta}^{[k]*};x_i,y_i)\bm{\Delta}_{\bm{\theta}^*}^{[k]}\bigg)
\\
& + I_{\{s>0\}}\bigg(-\frac{1}{\sqrt{\alpha_n}}\sum_{i:\tau^{[k]*}<x_i\le\tau^{[k]*}+s\alpha_n/n}\bm{\Delta}_{\bm{\theta}^*}^{[k]\T}g'(\bm{\theta}^{[k+1]*};x_i,y_i)
\\
& \ \phantom{+ I_{\{s>0\}}\bigg(} + \frac{1}{2\alpha_n}\sum_{i:\tau^{[k]*}<x_i\le\tau^{[k]*}+s\alpha_n/n}\bm{\Delta}_{\bm{\theta}^*}^{[k]\T}g''(\bm{\theta}^{[k+1]*};x_i,y_i)\bm{\Delta}_{\bm{\theta}^*}^{[k]}\bigg) + \oP(1).
\end{align*}

Now, letting $\{W_s\}_{s\in\mathbb{R}}$ be a two-sided standard Brownian motion with $\E[W_s]=0$ and $\V[W_s]=|s|$, and defining the limit of $(\bm{\Delta}_{\bm{\theta}^*}^{[k]\T}\E[-g''(\bm{\theta}^{[k]*};x,y)]\bm{\Delta}_{\bm{\theta}^*}^{[k]})^{1/2}$ and $(\bm{\Delta}_{\bm{\theta}^*}^{[k]\T}\allowbreak\E[-g''(\bm{\theta}^{[k+1]*};x,y)]\bm{\Delta}_{\bm{\theta}^*}^{[k]})^{1/2}$ as $\sigma^{[k]}$, it holds
\begin{align*}
& \frac{1}{\sqrt{\alpha_n}}\sum_{i:\tau^{[k]*}+s\alpha_n/n<x_i\le\tau^{[k]*}}\bm{\Delta}_{\bm{\theta}^*}^{[k]\T}g'(\bm{\theta}^{[k]*};x_i,y_i)\xrightarrow{\rm d}\sigma^{[k]}W_s, 
\\
& \frac{1}{\alpha_n}\sum_{i:\tau^{[k]*}+s\alpha_n/n<x_i\le\tau^{[k]*}}\bm{\Delta}_{\bm{\theta}^*}^{[k]\T}g''(\bm{\theta}^{[k]*};x_i,y_i)\bm{\Delta}_{\bm{\theta}^*}^{[k]}\xrightarrow{\rm p}-\sigma^{[k]2}|s|,
\\
& -\frac{1}{\sqrt{\alpha_n}}\sum_{i:\tau^{[k]*}<x_i\le\tau^{[k]*}+s\alpha_n/n}\bm{\Delta}_{\bm{\theta}^*}^{[k]\T}g'(\bm{\theta}^{[k+1]*};x_i,y_i)\xrightarrow{\rm d}\sigma^{[k]}W_s,
\\
& \frac{1}{\alpha_n}\sum_{i:\tau^{[k]*}<x_i\le\tau^{[k]*}+s\alpha_n/n}\bm{\Delta}_{\bm{\theta}^*}^{[k]\T}g''(\bm{\theta}^{[k+1]*};x_i,y_i)\bm{\Delta}_{\bm{\theta}^*}^{[k]}\xrightarrow{\rm p}-\sigma^{[k]2}|s|,
\end{align*}
and therefore we obtain 
\begin{align*}
Q^{[k]}(\tau^{[k]*}+s\alpha_n/n)\xrightarrow{\rm d}V^{[k]}(s)\equiv\left\{
\begin{array}{ll}
-\sigma^{[k]2}|s|/2+\sigma^{[k]}W_s & (s\le 0) \\
-\sigma^{[k]2}|s|/2+\sigma^{[k]}W_s & (s\ge 0)
\end{array}
\right..
\end{align*}
Letting $\tilde{V}^{[k]}(\cdot)$ be a copy of $V^{[k]}(\cdot)$, the same calculation as in \cite{Nin15} yields
\begin{align*}
\E\bigg[\sup_sV^{[k]}(s)\bigg] = \E\bigg[-\tilde{V}^{[k]}\bigg(\argsup_sV^{[k]}(s)\bigg)\bigg] = \frac{3}{2},
\end{align*}
which completes the proof.

\section*{Acknowledgments}

This research was supported by JSPS KAKENHI (23H00809 and 23K18471).

\bibliography{list}

\end{document}